\documentclass[journal]{IEEEtran}

\usepackage{graphicx}
\usepackage{stfloats}
\usepackage{algorithm}
\usepackage[noend]{algpseudocode}
\usepackage{booktabs}
\usepackage{tabularx}
\usepackage{subcaption}
\usepackage{multicol}
\usepackage{epstopdf}
\usepackage{amsmath}
\usepackage{amssymb}
\usepackage{amsthm}

\newtheorem{definition}{Definition}
\newtheorem{proposition}{Proposition}
\newtheorem{theorem}{Theorem}

\begin{document}

\title{Graph Permutation Entropy: Extensions to the Continuous Case, A step towards Ordinal Deep Learning, and More}

\author{
\IEEEauthorblockN{Om Roy\IEEEauthorrefmark{1},\IEEEauthorrefmark{2} Avalon Campbell-Cousins\IEEEauthorrefmark{1}, John Stewart Fabila Carrasco\IEEEauthorrefmark{1},\IEEEauthorrefmark{3} Mario A Parra\IEEEauthorrefmark{4}, Javier Escudero\IEEEauthorrefmark{1}\thanks{Corresponding author: Om Roy (email: o.roy.2022@uni.strath.ac.uk)}
}
\\
\bigskip
\IEEEauthorrefmark{1} School of Engineering, Institute for Imaging, Data and Communications, University of Edinburgh, Edinburgh, Scotland, United Kingdom

\IEEEauthorrefmark{2} Computer and Information Sciences, University of Strathclyde, Glasgow, Scotland, United Kingdom

\IEEEauthorrefmark{3} School of informatics, University of
Edinburgh, Edinburgh, Scotland, United Kingdom

\IEEEauthorrefmark{4} Psychological Sciences and Health, University of Strathclyde, Glasgow, Scotland,United Kingdom

}





\maketitle

\begin{abstract}
Nonlinear dynamics play an important role in the analysis of signals. A popular, readily interpretable nonlinear measure is Permutation Entropy. It has recently been extended for the analysis of graph signals, thus providing a framework for non-linear analysis of data sampled on irregular domains. Here, we introduce a continuous version of Permutation Entropy, extend it to the graph domain, and develop a ordinal activation function akin to the one of neural networks. This is a step towards Ordinal Deep Learning, a potentially effective and very recently posited concept. We also formally extend ordinal contrasts to the graph domain. Continuous versions of ordinal contrasts of length 3 are also introduced and their advantage is shown in experiments. We also integrate specific contrasts for the analysis of images and show that it generalizes well to the graph domain allowing a representation of images, represented as graph signals, in a plane similar to the entropy-complexity one. Applications to synthetic data, including fractal patterns and popular non-linear maps, and real-life MRI data show the validity of these novel extensions and potential benefits over the state of the art. By extending very recent concepts related to permutation entropy to the graph domain, we expect to accelerate the development of more graph-based entropy methods  to enable nonlinear analysis of a broader kind of data and establishing relationships with emerging ideas in data science. 

\end{abstract}

\textbf{
Understanding complex data requires advanced methods, especially for signals from irregular domains like graphs. This research introduces a continuous version of Permutation Entropy, adapting it for graph signals to enable nonlinear analysis. By developing an ordinal activation function and extending ordinal contrasts to graphs, the study advances the emerging field of Ordinal Deep Learning. Experiments on synthetic and real-life data across multiple disciplines demonstrate the effectiveness of these methods, promising enhanced data analysis and interpretability.
}

\IEEEpeerreviewmaketitle

\section{Introduction}

\IEEEPARstart{G}{raph} permutation entropy~\cite{Carrasco22} has recently extended the well studied permutation entropy (PE)~\cite{Bandt08} to graph signals. Given the rapid development of sensor measurement equipment and methods in recent years, graph signals have grown in popularity. By using the graph formalisation to consider spatial dependencies in data, it is possible to improve the analysis of various real-life applications, such as weather measurements and neuro-imaging data -- e.g., electroencephalogram (EEG) and MRI signals \cite{john24,Roy24}. This includes the analysis of signals not bounded by time such as images and fractal surfaces which can be improved by considering the spatial structure of such data \cite{Wang23}. 

 While classical, uni-variate PE \cite{Bandt08} considers sequential samples of time series when computing the ordinal patterns, permutation entropy for graph signals($PE_G$) extends this concept and considers $k$-step node neighbourhood sample averages \cite{Carrasco22}. In addition to allowing the extension of PE to graph signals, $PE_G$ allows the direct representation of classical PE as a special case of $PE_G$ by using an appropriate underlying graph. That is, $PE_G$ reduces to PE when considering that the time series has been sampled on a directed path graph  \cite{Carrasco22}. In addition, we can manipulate the underlying graph and consider topologies that best represent previously under-utilized dependencies in the signals, allowing for a truly multi-variate version of permutation entropy \cite{Carrasco22,mvperm}.\footnote{Of note, while other literature exists on `graph entropy', those methods use spectral properties of graph shift operators such as the Graph Laplacian or Adjacency Matrix Eigenvalues \cite{paeng22}. While informative, the information gained from such approaches is completely determined by the topology of the graph and \textit{not} the graph signals themselves.} Overall, the relevance of the graph-based entropy framework is also illustrated by the very recent implementation of bubble entropy for graph signals \cite{Dong24}, which used the aforementioned concept of $k$-step node neighbourhood sample averages, and the even more recent definition of other versions of PE building on nearest neighbours graphs \cite{knn}.

Permutation Entropy (PE) \cite{Bandt08}, while intuitive, robust and easy to understand, has some drawbacks. One of them is the lack of consideration of amplitude values in the computation of patterns and the handling of ties or equal values. While studies have addressed these limitations in uni-variate PE, this has yet to be considered in the graph domain. Zanin~\cite{Zanin23} in particular introduced continuous ordinal patterns and demonstrated the link between ordinal analysis and deep learning. This work generalized the ordinal pattern approach (i.e., the permutation patterns that form PE) by representing time series in terms of their similarity to a given representative pattern. This considered the amplitude values in the samples and allowed for the optimization of the ordinal pattern used to best represent the data under study, akin to the optimization procedure of traditional Machine Learning or Deep Learning methods.

Here, we extend the idea behind this methodology to the computation of a continuous version of $PE_G$. Our approach adapts a more natural method of describing continuous ordinal patterns and uses their mutual exclusivity as a way to compute a joint probability distribution. In particular, we introduce a ordinal activation function. This is any continuous and strictly monotonically increasing function that maps embedding vectors to a continuous domain while maintaining the ordinal structure of the patterns. This allows us to exploit certain amplitude-related features in windows of a given length depending on the choice of activation function. We show this is an improvement on the discrete case in real-life and synthetic data.

The literature has recently paid attention to pattern contrasts \cite{Bandt23} as a  method to analyse uni-variate time series. Namely, Bandt~\cite{Bandt23} introduces an orthogonal system of ordinal pattern contrasts which results from a linear combination of pattern frequencies. The statistical independence of these contrasts provides different types of information that can be used to detect certain signal properties. In particular, one of them, the turning rate contrast, was shown to be able to evaluate sleep depth directly from EEG data \cite{Bandt23}.

Here, we formally extend these contrasts to the graph domain while also introducing an amplitude-aware continuous version of these contrasts. We show how our continuous method improves on the traditional discrete version of the contrasts and how it generalizes effectively to graph signals.

The pattern contrasts are currently restrained to neighbourhoods of size 3, i.e., $6=3!$ possible values. This is due to increases in statistical complexity and computational costs of computing patterns of longer length \cite{Bona23}. However, the ideas behind these computations remain similar for longer length patterns. Bandt et al.~\cite{Bandt22} recently introduced two new ordinal patterns for the analysis of images. They took patterns of length 4 (24 possible patterns) where 3 `types' are computed based on symmetrical properties and statistical distributions of the parameters. Introducing an orthogonal coordinate system or basis for the analysis of images in a plane similar to the entropy-complexity plane.

Here, we extend these ideas to the graph domain and show that these processes integrate seamlessly and provide promising results in the graph signal analysis of images, in particular iso-tropic images which are shown to be rotation invariant \cite{Rib12}. This result in particular could be useful in machine learning computer vision applications where the ability to detect an image regardless of rotation is an important consideration and translates directly to the out-of distribution detection problem. We also show how this method translates to the analysis of fractal surfaces, a class of images that can be generated by a stochastic algorithm \cite{Wang23}.

\section{Methods}

\subsection{Background and Motivation}

We will first introduce some preliminary notation and information regarding Graph Signals and Graph Permutation Entropy.

A  graph \( G \) is defined as the triple \( G = (V, E, \textbf{A}) \) which consists of a finite set of vertices or nodes \( \textbf{v} = \{1, 2, 3, \ldots, N\} \), an edge set \( E \subset \{(i, j) : i, j \in V\} \), and \textbf{A} is the corresponding \( N \times N \) adjacency matrix with entries \(  \textbf{A}_{ij} \)=1 if \( (i, j) \in E \), and \( 0 \) otherwise.

A graph signal is a real function defined on the vertices, i.e., \( \textbf{x}: V \rightarrow \mathbb{R} \). The graph signal \( \textbf{x} \) can be represented as an \( N \)-dimensional column vector, \( \textbf{x} = [x_1, x_2, \ldots , x_N]^{\text{T}} \in \mathbb{R}^N \) (with the same indexing of the vertices).

Given a graph \( G = (V, E, \textbf{A}) \), we define a function on the vertices \( \text{deg}_k : V \rightarrow \mathbb{R} \) given by
\begin{equation}
\text{deg}_k(i) := \sum_{j \in V} (\textbf{A}^k)_{ij}
\end{equation}

for \( k = 1 \), denote \( \text{deg}_1(i) \) as \( \text{deg}(i) \), where \( \text{deg}(i) \) represents the degree of vertex \( i \), i.e., the number of edges incident to it.

Given a vertex \( i \), we define \( N_k(i) \) as the set of all vertices connected to the vertex \( i \) with a walk on \( k \) edges, i.e.,
\[
N_k(i) := \{ j \in \textbf{v} \mid \text{a walk exists joining } i \text{ and } j \text{ on } k \text{ edges} \}.
\]

The normalized Laplacian is defined using the adjacency matrix as follows:
\[
\Delta := \textbf{I} -\textbf{D} ^{-\frac{1}{2}} \textbf{A} \textbf{D}^{-\frac{1}{2}} ,
\]
where $\textbf{D}$ is the degree matrix, i.e., a diagonal matrix given by \( \textbf{D}_{ii} = \text{deg}(i) \).

Some important things to note is that the $k^{\text{th}}$ power of the adjacency matrix \textbf{A} that encodes the entire graph topology provides the frequency of the number of $k$-walks between two vertices. i.e., when $k=1$, it counts the number of edges between node $i$ and $j$.

To construct the embedding vector that is analogous to the permutation patterns in traditional PE for $m=2$, for each \( i \in V \), we define the pair \( y_i \) where its first component is the value of the signal $\textbf{x}$ at node \( i \) and the second component is the average of the signal $\textbf{x}$ on the neighbours of \( i \),
\[
\textbf{y}_i := (x_i, (\textbf{I} - \Delta) x_i) = \left( x_i,\textbf{D} ^{-\frac{1}{2}}\textbf{A}\textbf{D}^{-\frac{1}{2}} x_i \right) .
\]

The \( N \) pairs are ordered according to their relative values. This is easily extended for increasing values of $m$ by changing the power of the adjacency matrix in the construction of $(\textbf{I} - \Delta)$. This yields a distinct pattern. As an example, when $m=3$, we get the pattern `123' where the signal value of the nodes is less than the average of the signal values for its 1 step and 2 step neighbors. We are now ready to define  Permutation Entropy  for Graph Signals ($PE_G$).

%\vspace{10pt} 

\subsection{Permutation Entropy for Graph Signals}

\begin{definition}[Permutation Entropy for Graph Signals]

%\vspace{5pt} 

Let $G = (V, E, \textbf{A})$ be a graph and $\textbf{x} = \{x_i\}_{i=1}^{N}$ be a signal on the graph. The permutation entropy for graph signals, $PE_G$, is defined as follows \cite{Carrasco22}:

%\vspace{5pt} 

\begin{enumerate}
    \item For  $2 \leq m \in \mathbb{N}$ the embedding dimension, $L \in \mathbb{N}$ the delay time, and for all $i = 1, 2, \dots, n$, we define
    \begin{equation}
\textbf{y}^{kL}_i = \frac{1}{|N^{kL}(i)|} \sum_{j \in N^{kL}(i)} x_j 
= \frac{1}{|N^{kL}(i)|} (\textbf{A}^{kL}\textbf{x})_i,
\end{equation}
    where $N^k(i)$ is defined as the set of vertices \( j \in V \) such that there exists a walk on \( k \) edges joining \( i \) and \( j \).
    
    \item The embedding vector is then given by
    \begin{equation}
    y^{m,L}_i = \left(y^{kL}_i\right)_{k=0}^{m-1} = \left(y^0_i, y^L_i, \dots, y^{(m-1)L}_i\right).
    \end{equation}
       % \vspace{5pt}

        \item The embedding vector $\textbf{y}^{m,L}_i$ is sorted to be in increasing order.
    %    \vspace{5pt}

    \item The relative frequency for the distinct permutations $\pi_1, \pi_2, \dots, \pi_k$, where $k=m!$, is denoted by $p(\pi_1), p(\pi_2), \dots, p(\pi_k)$. The permutation entropy for the graph signal $\textbf{x}$ is then computed by using the normalized Shannon entropy:
       % \vspace{5pt}

    \begin{equation}
    \text{$PE_G$}(m, L) = - \frac{1}{\ln m!} \sum_{i=1}^{k} p(\pi_i) \ln p(\pi_i).
    \label{$PE_G$}
\end{equation}

\end{enumerate}
\end{definition}
   % \vspace{5pt}

    %\vspace{5pt}

\subsection{Ordinal Contrasts for Graph Signals}

    \vspace{5pt}

Given the embedding vector when $m=3$:

\begin{equation}
    \textbf{y}^{L}_i = \left(y^{kL}_i\right)_{k=0}^{2} = \left(y^0_i, y^L_i,y^{2L}_i,\right),
    \end{equation}
ordinal Contrasts for $m=3$, introduced by Brandt in \cite{Bandt23} for uni-variate time series, naturally extend to the graph signals.
The values are sorted to give $3!=6$ distinct patterns.  By considering the fact that the sum of all pattern frequencies must be equal to one, we get a joint probability distribution spanned by the individual pattern distributions. Also, a key difference to the uni-variate case is the fact that we can now assign a pattern to \textit{each} node. This granularity allows us a level of precision previously unattainable.
We define the permutation contrasts with the graph signal properties they aim to emphasize as in Table \ref{Tab:Def}.

\begin{table}[H]
    \centering
 
    \caption{Definition of Ordinal Contrasts}
    \begin{tabularx}{\linewidth}{|c|X|X|}
    \hline
    Contrast & Equation & Property \\
    \hline
    $\alpha$ & $\alpha = p(132) + p(213) + p(231) + p(312)$ & Turning rate \\
    $\beta$ & $\beta = p(123) - p(321)$ & Up-down balance \\
    $\tau$ & $\tau = p(123) + p(321) - \frac{1}{3}$ & Persistence \\
    $\gamma$ & $\gamma = p(213) + p(231) - p(132) - p(312)$ & Rotational symmetry \\
    $\delta$ & $\delta = p(132) + p(213) - p(231) - p(312)$ & Up-down scaling \\
    \hline
    \end{tabularx}
       \label{Tab:Def}
\end{table}

While it is unusual for ordinal patterns to have representative meanings, `contrasts' (analogous to the concept of variance in traditional statistics) \cite{Bandt23} emphasize certain graph signal properties providing discriminate information. Furthermore, our graph version of permutation contrasts, using its node neighborhood method, considers both the data and the graph topology which adds additional value to the information provided by the contrasts.

\begin{figure}[H]
    \centering
    \includegraphics[width=0.4\textwidth]{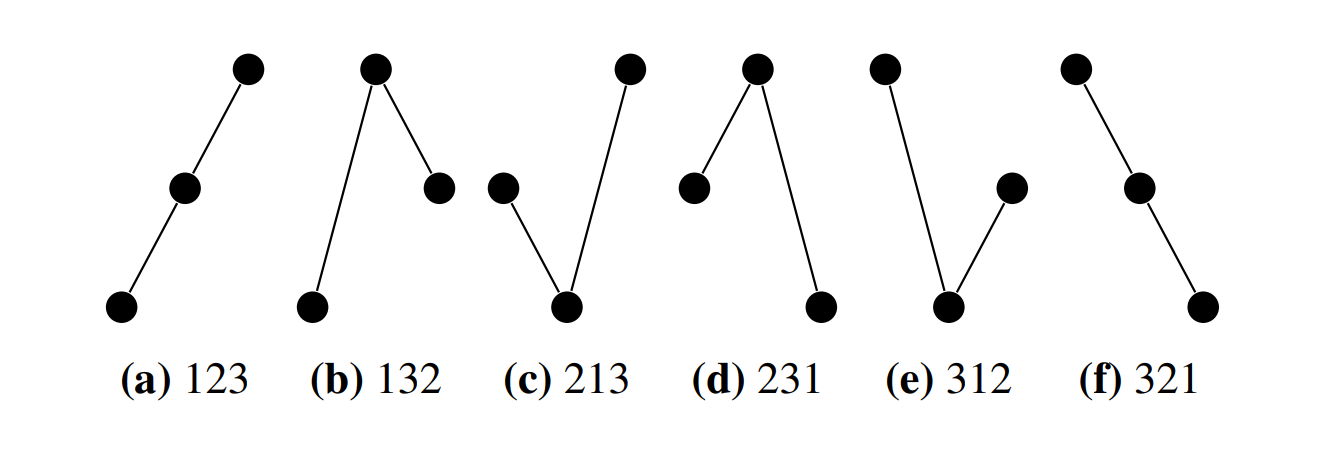}
    \caption{Permutation patterns for $m=3$ from \cite{john24}.}
    \label{fig:example}
\end{figure}

For example, the parameter $\alpha$ has the `Up-Down balance' property as it represents the frequency of turning points relative to monotonic behaviour in the signal. This is clear visually as we see the patterns 132, 213, 231 and 312 represent clear minima and maxima with the middle value always being a 1 (minima) or a 3 (maxima). Smooth signals will have lower values of $\alpha$, while highly oscillating signals will have higher values. This behaviour is somewhat similar to the eigenvalues of the Graph Laplacian where the higher eigenvalues correspond to higher frequencies. Thus, it would make sense for the Laplacian eigenvalues to have a positive correlation with $\alpha$, i.e., signals defined on a graph with generally large Laplacian eigenvalues would have a higher value for $\alpha$ regardless of the actual behaviour of the signal itself. This relates information from the graph topology (Laplacian eigenvalues) and the graph signal itself ($\alpha$), showing how both topological and signal properties influence the Graph Permutation contrasts compared to the uni-variate case.

For a detailed explanation on the interpretation of the other parameters, we refer the interested reader to \cite{Bandt23}. For this work, brief explanations will suffice. $\beta$ represents the up-down balance, i.e., the `balance' of strictly increasing and decreasing patterns in the signal.  $\gamma$ distinguishes patterns with the intermediate value at the beginning compared to the end of the length 3 pattern (thus the name rotational symmetry), with increasing values for high amplitude oscillations, negative for damped oscillations and zero for time-reversible processes. $\delta$, termed as `down-scaling', is highly correlated to $\beta$, and is zero for time-reversible, self-similar and spatially symmetric processes. It is quite evident that $\beta$ , $\gamma$ and $\delta$ all represent some sort of `distance' from symmetry and/or reversibility.

\subsection{Continuous ordinal patterns}
We will now introduce continuous ordinal patterns from Zanin~\cite{Zanin23} as a necessary pre-requisite for our introduction of continuous Graph Permutation Entropy and its corresponding contrasts. Defined in a natural way, an ordinal pattern of length three is normalized to the range $\left[-1,1\right]$ where the intermediate pattern can be any value in the range (infinite possible values exist). The possible patterns are:\begin{align*}
(-1, *, 1) \quad & (-1, 1, *) \quad (\ast, -1, 1) \\
(\ast, 1, -1) \quad & (1, -1, *)  \quad (1, *, -1)
\end{align*}
with $\ast$ representing the continuous intermediate value. Given a reference pattern as input and a time-series, windows of length $m$ are used to find the normalized distance between each window and the reference pattern. For instance, given a window starting at time $t$  where $s=(x_t, x_{t+1}...,x_{t+m-1})$
is the normalized window of length $m$ and a reference pattern some permutation of i.e. $m=3$, $\pi = (\pi_0, \pi_1, \pi_2)$, with $\pi_i = -1$ and $\pi_j = 1$ (where $i$ and $j \in \{0, 1, 2\}$).
The `similarity' of the reference pattern with a sub-window $\phi_{\pi}(t)$ is defined as:
\[
\phi_{\pi}(t) = \frac{1}{2D}\sum_{i=0}^{D-1} d_i = \frac{1}{2D}\sum_{i=0}^{D-1} |\pi_i - s^*_i|,
\]
with the $1/2D$ constant normalizing the value between 0 and 1 to give a natural representation of similarity, such as when they are perfectly equal $\phi_{\pi}(t)$ is 0. The overall similarity of the reference pattern to the signal as whole can be calculated by taking the average of $1-\phi_{\pi}(t)$ over the windows. This shows how important a reference pattern is in determining the dynamics of the signal. Note the intermediate value behaves like a tunable parameter which can be optimized to find the `best' pattern for the signal at hand. Additionally, it is interesting to note that iteratively applying the reference pattern to each window returns another time-series, akin to a convolution operation. This can be considered as a sort of `filtration' of the time-series and thus the optimization problem depends on choosing the best or most optimal filter for the signal.

This is easily extended to the Graph Signal case where instead of taking windows of length $m$, we take $m,m-1...$-step neighborhood averages and essentially replicate the process. Thinking of this in graph signal terms, we can now apply an individual, potentially distinct, filter to each node and, in principle, find the optimal filter at the  node level in order to optimize some sort of global or local criterion, providing a more granular form of optimization. Zanin~\cite{Zanin23} applied an interesting optimization technique to obtain a reference pattern that maximizes a Kolomorogov-Smirnov (K-S) two-sample test to find a continuous pattern for which its average similarity ($1-\phi_{\pi}(t)$) is greater than that of a randomly shuffled time series. Maximizing the K-S statistic there meant obtaining the reference pattern that best identified the temporal dependencies of the dynamics that were being analyzed. This does imply that the reference pattern can be optimized to distinguish different characteristics of the signal, e.g., robustness to noise. This approach is intuitive and we will use this as our main inspiration for the Continuous version of (Graph) Permutation Entropy.

\subsection{Continuous (Graph)
Permutation Entropy and Ordinal Contrasts}

We first define our ordinal activation function. This is a function that adds non-linearity to the embedding vector in Graph Permutation Entropy. This corresponds to the traditional deep learning activation function. Note this can be applied to traditional Permutation Entropy analogously.

\vspace{5pt} 

\begin{definition}[Ordinal activation function]

Let $\mathbb{R}^m$ denote the space of real-valued embedding vectors of dimension $m$. Then, $\rho_{\text{act}}: \mathbb{R}^m \rightarrow \mathbb{R}^m$ is an ordinal activation function if it satisfies the following properties:
\begin{enumerate}
    \item Strict increasing monotonicity: For any $\mathbf{x}, \mathbf{y} \in \mathbb{R}^m$ such that $\mathbf{x} < \mathbf{y}$ (component-wise), then $\rho_{\text{act}}(\mathbf{x}) < \rho_{\text{act}}(\mathbf{y})$.
    \item Continuity: The function $\rho_{\text{act}}$ is continuous on its domain $\mathbb{R}^m$.
\end{enumerate}

Moreover, for a real-valued embedding vector $\mathbf{x} = (x_1, x_2, \ldots, x_m) \in \mathbb{R}^m$, the ordinal activation function $\rho_{\text{act}}$ maps it to another real-valued embedding vector $\mathbf{y} = (\rho_{\text{act}}(x_1), \rho_{\text{act}}(x_2), \ldots, \rho_{\text{act}}(x_m)) \in \mathbb{R}^m$.

\end{definition}

\vspace{10pt} 

Now we will introduce a continuous version of Graph Permutation Entropy and its natural extension to its corresponding length 3 pattern contrasts. 

\vspace{5pt} 

\begin{definition}[Continuous (Graph) Permutation Entropy] \label{def:$CPE_G$}

Let $G = (V, E, \textbf{A})$ be a graph and $\textbf{x} = \{x_i\}_{i=1}^{N}$ be a signal on the graph. The Continuous Permutation entropy for the graph signals, $CPE_G$, is calculated as follows with embedding dimension $m$ and delay $L$: 

\vspace{5pt} 
\begin{enumerate}
    \item We (optionally) normalize the entire graph signal over the topology using, for example, the standard $z$-score formula (or a suitable normalization method specific to the data at hand) and compute the embedding vector \( \textbf{y}^{m,L}_i \) by calculating the \( m \)-neighborhood of a given node \( i \) using equation (2):
    \begin{align*}
    %\\
    \textbf{y}^{m,L}_i &= \left(y^{kL}_i\right)_{k=0}^{m-1} = \left(y^0_i, y^L_i, \dots, y^{(m-1)L}_i\right).
    \end{align*}
    
    \item Apply the ordinal activation function ($\rho_\text{act}$). 
    \begin{align*}
    \textbf{a}^{m,L}_i &= \rho_{\text{act}}(\textbf{y}^{m,L}_i ).
    \end{align*}
    Note that since the ordering is maintained, distinct permutations are also maintained.
    
    \item For each permutation of length \( m \), we store the corresponding continuous ordinal activated vector \( a \) in the respective pattern class. The new embedding vector is now:
    \begin{equation}
    \begin{aligned}
    \textbf{a}^{m,L}_i &= \left(a^{kL}_i\right)_{i=0}^{m-1} = \left(a^0_i, a^L_i, \dots, a^{(m-1)L}_i\right), \\
     \end{aligned}
    \end{equation}
     The embedding vector (if the limit exists) is bounded as:
    \begin{equation}
    \begin{aligned}
    \inf \textbf{a}^{m,L}_i &= \lim_{x \to -\infty} \rho_{\text{act}}(x), \\
    \sup \textbf{a}^{m,L}_i &= \lim_{x \to +\infty} \rho_{\text{act}}(x).
    \end{aligned}
    \end{equation}
    
    \item Now define \( l \) as the many-to-one function that maps each node  to a  permutation:
    \begin{equation}
    l : \mathbb{R}^m \rightarrow \mathbb{S}_m
\end{equation}
    where $\mathbb{S}_m$ is the set of permutations of length $m$.
    
    \item Define \( T \) as a function that maps a  statistic to the activated embedding vector (such as the variance or mean), and take the absolute value of the continuous statistic to compute the \( * \) value for the permutation \( p \):

 \begin{equation}
    T : \mathbb{S}_m \rightarrow \mathbb{R}
\end{equation}
    
    \begin{equation}
    *_{i,p} = \left| (T \circ l)(\textbf{a}^{m,L}_i) \right|.
    \end{equation}

   \item Compute \( *_{i,p} \) for each node \( i \). Every node is assigned a pattern and a real value. For every permutation in \( \textbf{a}^{m,L}_i \), let \( Z^*_p \) denote the sum of the \( * \) values for nodes with permutation \( p \), where \( p \) ranges from 1 to \( m! \). Calculate the total sum of \( * \) values and then compute \( Z^*_p \) for each permutation. The probability of \( * \) originating from permutation \( p \), denoted as \( P(X^* = p) \), is calculated as:
    \[
    P(X^* = p) = \frac{Z^*_p}{\sum_{p=1}^{m!} Z^*_p}
    \]
    or in integral form:
    \[
    P(X^* = p) = \int_{\mathcal{X}^*} P(X^*) \, d(*_p)
    \]
    where \( \mathcal{X}^* \) is the support of the continuous joint probability distribution \( P(X^*) \) of all continuous ordinal patterns.

  \item Finally, use the definition of (normalized) entropy as follows:
\begin{equation}
\text{$CPE_G$}(m, L) = - \frac{1}{\ln m!} \sum_{p=1}^{m!} P(X^* = p) \ln P(X^* = p).
\label{$CPE_G$} 
\end{equation}
or:
\begin{equation}
\text{$CPE_G$}(m, L) = - \frac{1}{\ln m!} \int_{\mathcal{X}^*} P(X^*) \ln P(X^*) \, dX^*
\end{equation}

\end{enumerate}

\end{definition}

 Essentially, we are taking $P(X^* = p)$ to be the proportion of $P(X^*)$ composed  from the continuous * values from the individual distribution $P(X^*_p)$.

 This leads to the following proposition which allows us to perform such an operation:

\vspace{5pt} 
\begin{proposition}
\vspace{5pt} 

\textit{The individual probability distributions of continuous * values of nodes with ordinal pattern $p$, $P(X^{*}_{p})$  are mutually exclusive, and thus the joint probability distribution $P(X^{*})$ is well-defined.}
\vspace{5pt} 
\end{proposition}

\begin{proof}
Let $i$ be a node in the graph. Suppose $p_1$ and $p_2$ are two distinct permutations such that the distributions $P(X^{*}_{p1})$ and $P(X^{*}_{p2})$ have a common node $i$.

\vspace{5pt} 

Let $l(i) = p_1$ and $l(i) = p_2$.
Since $l$ is a function, $p_1$ and $p_2$ must be the same permutation, i.e., 

\[p_1 = p_2\]

This contradicts the assumption of distinct permutations. Therefore, each node can only be mapped to one permutation distribution.

\[\Rightarrow P(X^{*}_{p1}) \cap P(X^{*}_{p2}) = \emptyset\].

$\Rightarrow$  All $P(X^{*}_{p})$ distributions are mutually exclusive.

\[
\Rightarrow \sum_{p=1}^{m!} P(X^{*} = p) = 1
\]
and
\[
\Rightarrow\int_{\mathcal{X}^*} P(X^*) \, dX^*=1
\]
where $X^{*}$ is the joint probability distribution of all permutations $p$.

\end{proof}
\vspace{10pt}

This is an intuitive extension to normal $PE_G$ as now we are 
applying non-linearity to the window under consideration. The choice of function being restricted to continuous, strictly monotonically increasing functions allows us to maintain the ordering in patterns, this is an application of the \textit{invariance property} of Permutation Entropy with respect to monotonic transformation.

\begin{theorem}
\textit{The invariance property \cite{Bandt08} with respect to \textit{strictly} monotonic transformation of the time signal is an important property of the permutation entropy (PE). If $\textbf{x}$ is a time series, and $f$ is an arbitrary strictly increasing (or decreasing) real function, then the classical PE of the time series $\textbf{x}$ and $f(\textbf{x})$ are equal, i.e}
\begin{equation}
H(\textbf{x}) = H(f(\textbf{x}))
\end{equation}
\textit{where
$H(\textbf{x})$ is the permutation entropy of the original time series $X$,
$f(\textbf{x})$ represents the time series transformed by the function $f$, and
$H(f(\textbf{x}))$ is the permutation entropy of the transformed time series.}
\end{theorem}

As this property translates trivially to $PE_G$. It is implied that the application of our ordinal activation function on windows of length $m$ preserves the Discrete $PE_G$ of the signal in all cases.

This gives us the following proposition.

\begin{proposition}
    
\vspace{5pt} 

\textit{For all graph signals \( \textbf{x} \) over the graph \( G \), $PE_G$ is a special case of $CPE_G$ when \( T \) is equal to the identity function that maps all vectors to the scalar value 1. This property is maintained regardless of the choice of \( \rho_{\text{act}} \).}

\end{proposition}

\begin{proof}
Define the set \( \mathcal{M} \) as the set of real-valued functions that are strictly monotonically increasing. Define:
\[
I_m : \mathbb{R}^m \rightarrow \mathbb{R}
\]
\[
I_m(y) = 1
\]

Fix:
\[ 
T = I_m 
\]

As:
\[
\rho_\text{act} \in \mathcal{M}
\]
\[
y_i > y_j \Rightarrow \rho_{\text{act}}(y_i) > \rho_{\text{act}}(y_j)
\]

This maintains the ordering of \( (\textbf{y}_i^m,L) \). Thus note:
\[
|(T \circ l)(\rho_{\text{act}}(\textbf{y}_i^m,L))| = 1, \forall \rho_{\text{act}} \in \mathcal{M}
\]

\[
\Rightarrow *_{i,p} = 1 \text{ for each node mapped to a pattern } p
\]

\[
\Rightarrow 
Z^*_p = \pi_{p}
\]
where \(\pi_{p} \) is the frequency of pattern \( p \) in the signal.
\[
\Rightarrow P(X^* = p) = \frac{\pi_{p}}{\sum_{p=1}^{m!} \pi_{p}} = p(\pi_{p})
\]
which is just the relative frequency of counts of the patterns in Definition 1. Thus:
\[
T = I_m \Rightarrow CPE_G=PE_G , \forall \rho_{\text{act}} \in \mathcal{M}
 \]
\

\end{proof}

As a possible pre-processing step, $z$-score normalization can be used over the whole graph signal before constructing the embedding vector to account for the case where signal values are extremely small or  extremely large. This can result in every value in the signal being mapped to the same value due to the structure of certain ordinal activation functions. Using the hyperbolic tangent as an example with the range $[-1,1]$, very large values would be mapped to 1 and vice-versa for very small values. Now everything will be mapped to the class accounting for ties making the method impractical. This is prominent in some signal types, such as fMRI signals where signal values are typically very large. Interestingly, this can be seen as a direct ordinal deep learning translation of the vanishing/exploding gradient problem. (Optional) Normalization over the neighbourhood topology ensures that while the range of the signal is still technically infinite, extreme values are very unlikely to occur.  The normalization is typically done over the whole signal $\textbf{x}$ at the start so relative amplitudes are maintained in the windows. Once again, this is similar to deep learning methods where normalization is sometimes used as a pre-processing step \cite{norm}. We can use any type of normalization such as min-max normalization as well depending on the data at hand. As in the discrete case, each node has a continuous pattern maintaining the granular analysis efficiency. Ties are very rare due to how the embedding vector is constructed by taking neighbourhood averages, the continuous ordinal activation function being \textit{strictly} monotonically increasing further reduces the likelihood of ties, in-fact in our experiments we did not encounter any ties.

The use of a  statistic ($T$) allows us to extract information from the activated window that we deem relevant. For example, we may use the variance when we think the variation in amplitude may have discriminating information, or the mean when the entire window may be informative. Using the median as $T$ for a non-negative monotonically increasing ordinal activation function returns the intermediate value method used by Zanin~\cite{Zanin23}, while using the rank function (known as tiedrank in MATLAB) as our ordinal activation function with the identity function (maps all vectors to a scalar value of 1), we obtain classical $PE_G$. Given that these operations are consistent over all windows, the following step to compute the relative weighted probabilities and thus the probability distribution for input into Shannon Entropy is a valid operation. In essence, we are projecting all of the embedding vectors from the real line to the same ordinal activation function and comparing their relative behaviour on this continuous function.

We will show this with an example. When $m = 3$, the possible continuous permutation patterns are defined as:
\[
\begin{aligned}
& (1, 2, 3) && (1, 3, 2)  \\
& (2, 3, 1) && (3, 1, 2) \\
& (3, 2, 1)&& (2, 1, 3)
\end{aligned}
\]

Consider an embedding vector where $m=3$: $(0.5, 0.1, 1.7)$. Now let us choose a ordinal activation function. For example, the GELU (Gaussian Error Linear Unit) \cite{gelu} function, a commonly used activation function in many deep learning use cases, is defined as:
\begin{equation}
\mathrm{GELU}(x) = \frac{1}{2}x\left(1 + \mathrm{erf}\left(\frac{x}{\sqrt{2}}\right)\right)
\end{equation}
with $\mathrm{erf}$ being the traditional Gaussian error function.
Applying this function gives us $(0.3457,0.0540,1.6242)$. Now we choose the mean to be our  statistic ($T$) and we compute $\left|T(a^{m,L}_i)|\right|$. This gives us 0.6747 as our $*_{i,p}$  value for that node. We must take the absolute value as having negative weights in our probability computation will result in incorrect results. This can be considered as a weighted `count' as in discrete $PE_G$ we would have the value just being 1, we have 0.2282 here, determined by the non-linearity of the ordinal activation function and the corresponding  statistic. These are parameters that can be fine-tuned based on a priori information about the signal under consideration. Note how throughout the computation the order (213) of the embedding vector is maintained.

Given another (312) vector $(2.7,0.1,1.8)$, this is activated to $(2.691,0.054,1.735)$. Taking the  statistic gives us $*_{i,p}=1.493$. Observe how amplitude values are now taken into account while maintaining the ordinal pattern structure. Each * is unique to a pattern; it is not a global variable. 

Now, let's say these two were the only patterns. 

\[
\begin{aligned}
P(X^* &= 312) &= \frac{1.493}{1.493+0.6747} = 0.6887 \\
P(X^* &= 213) &= \frac{0.6747}{1.493+0.6747} = 0.3113
\end{aligned}
\]
where $P(X^* = 31*)+P(X^* = *13)=1$. We can now plug these probabilities into the Shannon Entropy formula.
Note that in the discrete case, if these were the only patterns, the distribution would be equal, as we are using relative frequencies, i.e., $p(312) = p(213) = 0.5$.

Observe that now the relative pattern probability distributions are being changed and mapped in a non-linear way while still maintaining the nature of the ordinal patterns. The higher amplitude pattern windows are being mapped to higher values at a rate determined by the ordinal activation function.
Another way to interpret this is that instead of measuring the proportion each ordinal pattern makes up of the total number of patterns we are measuring how much of the total power of the (graph) signal each ordinal pattern is contributing to.

Theoretically, we can now manipulate the parameters of this algorithm to `train' the model on a given signal or learn the best ordinal activation function specific to the signal.
For example, if we know that most values in the signal fluctuate around a certain range but relevant signal spikes are also not too different from this general background amplitude, we could use an exponentially increasing function to emphasize larger values or spikes in the neighbourhood as illustrated in Figure \ref{fig:ar1_comparison}. We can also now optimize the gradient of the ordinal activation function as an iterative training step similar to back-propagation. Furthermore, not only are amplitudes considered within windows but relative to the whole signal as well, addressing a long-running drawback of permutation entropy. The ordinal activation function can also be optimized to reduce the influence of noise by mapping to a function with a very slowly increasing gradient maintaining a stable value. This can clearly be extended to uni-variate and multi-variate time-series signals.

\begin{figure}[h]
    \centering
    \includegraphics[width=0.5\textwidth]{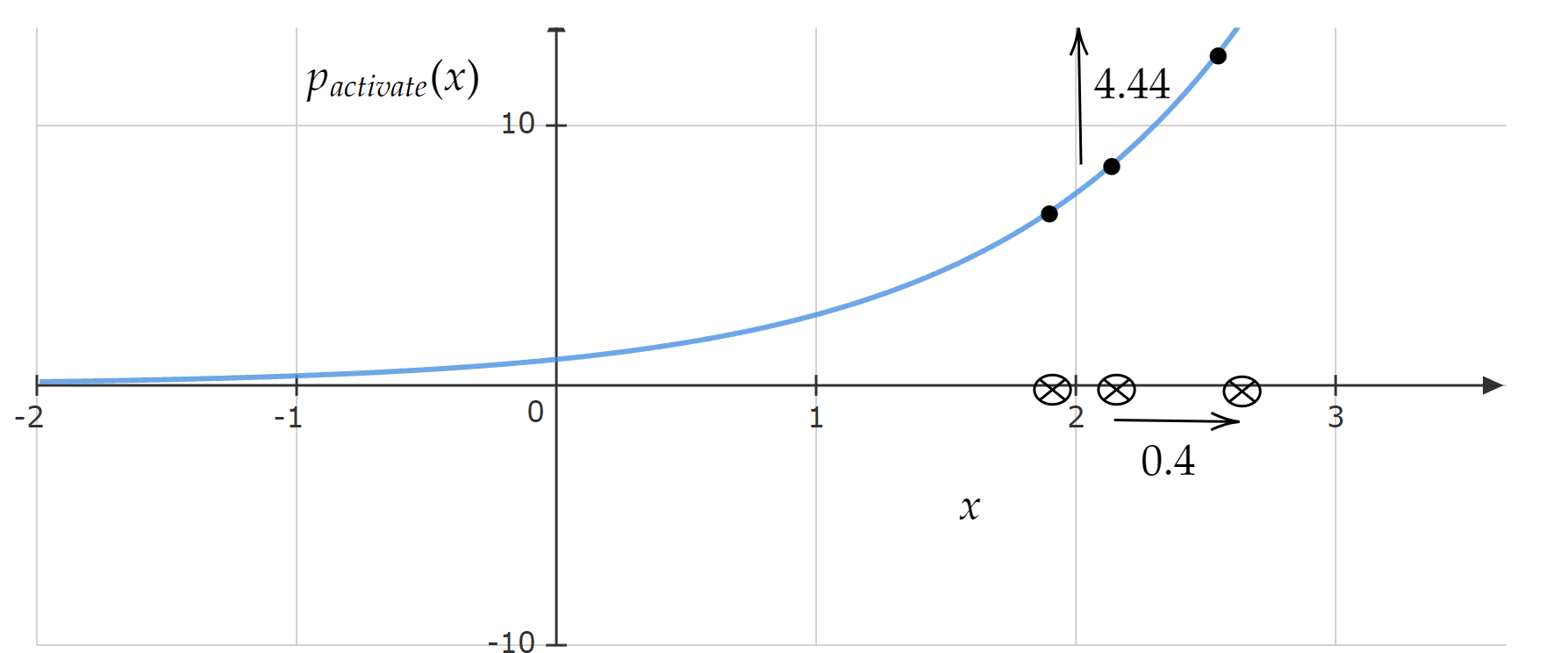}
    \caption{Example of spike emphasis when using exp ordinal activation function}
    \label{fig:ar1_comparison}
\end{figure}

\begin{definition}[Continuous (Graph) Ordinal Contrasts]
Analogously we define the continuous versions of the Graph Permutation Contrasts as in Table \ref{tab:Def2}

\begin{table}[H]
    
    \centering
    \caption{Definition of continuous ordinal contrasts.}
    \small
    
    \begin{tabular}{p{6cm}}
    \toprule
    Continuous Graph Contrasts \\
    \midrule
    $\alpha=P(X^* = 132) + P(X^* = 213) + P(X^* = 231) + P(X^* = 312)$ \\
    $\beta=P(X^* = 123) - P(X^* = 321)$ \\
    $\tau=P(X^* = 123) + P(X^* = 321) - \frac{1}{3}$ \\
    $\gamma=P(X^* = 213) + P(X^* = 231) - P(X^* = 132) - P(X^* = 312)$ \\
    $\rho=P(X^* = 132) + P(X^* = 213) - P(X^* = 231) - P(X^* = 312)$ \\
    \bottomrule
    \end{tabular}
    \label{tab:Def2}
\end{table}
\end{definition}

This is a straightforward extension as the probabilities can just be considered ordinal activation function weighted relative pattern frequencies compared to the relative pattern frequencies in the discrete case.

\subsection{Graph image pattern contrasts}

Methods to evaluate images are in high demand in recent times. \cite{Bandt22} recently introduced two new ordinal patterns for analyzing images. They use $2 \times 2$ windows to determine local patterns in images, the 24 possible permutations are then grouped into 3 types based on symmetrical properties while defining two new contrasts based on these types. We now extend this to the graph case and introduce the corresponding image contrasts. 
 %\vspace{5pt}
\begin{definition}[Graph Image Types]

\end{definition}

Let \(\{\textbf{X}_{i,j}\}_{i,j=1}^N\) represent an image with size \(N \times N\). Let \(G\) be a Regular 2D Grid graph of size \(N\). The adjacency matrix\textbf{ A} of the (undirected) grid graph is defined as:

\[
\textbf{A}_{(i,j),(k,l)} = 
\begin{cases} 
1 & \text{if } |i - k| + |j - l| = 1 \\
0 & \text{otherwise}
\end{cases}
\]

The adjacency matrix for a \(3 \times 3\) grid graph is:

\[
\begin{bmatrix}
0 & 1 & 0 & 1 & 0 & 0 & 0 & 0 & 0 \\
1 & 0 & 1 & 0 & 1 & 0 & 0 & 0 & 0 \\
0 & 1 & 0 & 0 & 0 & 1 & 0 & 0 & 0 \\
1 & 0 & 0 & 0 & 1 & 0 & 1 & 0 & 0 \\
0 & 1 & 0 & 1 & 0 & 1 & 0 & 1 & 0 \\
0 & 0 & 1 & 0 & 1 & 0 & 0 & 0 & 1 \\
0 & 0 & 0 & 1 & 0 & 0 & 0 & 1 & 0 \\
0 & 0 & 0 & 0 & 1 & 0 & 1 & 0 & 1 \\
0 & 0 & 0 & 0 & 0 & 1 & 0 & 1 & 0 \\
\end{bmatrix}
\]

We then define the Graph image patterns using Definition 1 with $m=4$,
\begin{equation}
    \textbf{y}^{4,L}_i = \left(y^{kL}_i\right)_{k=0}^{3} = \left(y^0_i, y^L_i, \dots, y^{(3)L}_i\right).
\end{equation}

Given a graph permutation pattern $(p_1,p_2,p_3,p_4)$, of a given node produced by sorting the embedding vector, the following algorithm assigns a type to the $24$ possible length $4$ patterns according to the algorithm:

\begin{algorithm}[H]
    \caption{Assignment of types of patterns for images for $m=4$}
    \begin{algorithmic}[1]
        \Procedure{AssignType}{$p_1$, $p_2$, $p_3$, $p_4$}
            \State Compute $N=$ number of unique values in $\{p_1, p_2, p_3, p_4\}$
            \If{$N = 1$ or $N = 0$}
                \State Assign $t$ randomly from $\{1, 2, 3\}$
            \ElsIf{$N = 2$}
                \If{$p_1 = p_2$ or $p_3 = p_4$}
                    \State Assign $t$ randomly from $\{2, 3\}$
                \Else
                    \State Assign $t$ randomly from $\{1, 2\}$
                \EndIf
            \Else
                \State Compute $a = (p_1 < p_2) + (p_3 < p_4)$
                \If{$a = 2$}
                    \State $a \gets 0$
                \EndIf
                \State Compute $b = (p_1 < p_3) + (p_2 < p_4)$
                \If{$b = 2$}
                    \State $b \gets 0$
                \EndIf
                \State Compute $t = a + b + 1$
            \EndIf
            \State \textbf{return} $t$
        \EndProcedure
    \end{algorithmic}
\end{algorithm}

While ambiguous at first, the grouping is somewhat intuitive and is directly translated to the graph domain by considering neighbourhood averages  of nodes compared to the neighbouring pixel values. The three types are formed in such a way that they describe the symmetry classes of $2 \times 2$ patterns compared to the symmetry group of a square. Analogously, the graph versions takes quadrilateral neighborhood averages which results in 4 values. The signal values are considered the `top-left' value in a $2 \times 2$ grid, two one step neighbor average is the `top-right entry', the 2 step neighbour averages is the `bottom left' entry and the 3 step neighbor average is the bottom right entry. In this case, each $2 \times 2$ square represents the local dependencies of each pixel in the image giving us a pattern for each pixel. The symmetric behaviour of the neighborhood averages of the nodes/pixels then determines the grouping into a given type. i.e., the type of a $2 \times 2$ pattern is the rank number, which shares
a diagonal with the average of the 3 step neighbourhood. As an example, if p3 is on the same diagonal as p4, we assign type 3. As such, the graph permutation pattern is in the sorted form (p3,p1,p2,p4) or (p3,p2,p1,p4). Intuitively, type 1 values are either decreasing or increasing the neighborhood $2 \times 2$ matrix we have defined. This can be considered to represent the smoothness of the graph signal.
Type 2 implies either the rows are increasing or decreasing or the columns are with the non parallel increase axis having one increase and one decrease. This usually occurs in tree-like structured images. In type 3, both values on one diagonal are larger than both values on another diagonal. This could represent an edge in an image in the diagonal direction. We have essentially used a form of quantization to compress 24 patterns into 3 types as in \cite{Bandt22}. This is computationally effective and preserves the important symmetrical features in the images. 
 %\vspace{10pt}
 
\begin{definition}[Graph Image Pattern Contrasts]
We now introduce the Pattern contrasts for images. Let $t_1$, $t_2$, $t_3$ represent the relative frequencies of each type in the image. We then define:
\[
\theta = t_1 - \frac{1}{3} \quad \text{and} \quad \kappa = t_2 - t_3.
\]

$\theta \in [-1/3,2/3]$ is defined as the smoothness parameter. It is similar to the persistence in the $m=3$ case and achieves higher values for monotonically increasing functions and lower values for highly oscillating ones. We can expect it to be directly correlated with the eigenvalues of the underlying graph of the image. Thus, we can manipulate or emphasize specific features of the image by changing the underlying graph (i.e., adding weights). $\kappa \in[-1,1]$ is unique to $m=4$, it quantifies the presence of branching structures in the image compared to noise (usually demonstrated by a `checkerboard' like behaviour or alternating sequences). These contrasts represent different information as shown in \cite{Bandt22} and can be considered akin to an analysis of an image in the entropy-complexity plane.
\end{definition}

The ties are treated methodologically and clearly, as the type classification is based on comparing horizontal and vertical values on the $2 \times 2$ transformed window. They are treated as in Table~\ref{tab:ties}.

\begin{table}[H]
\label{tab:ties}
\centering
\caption{Treatment of ties for $2\times 2$ patterns.}
\begin{tabular}{lcc}
\toprule
\textbf{Case} & \textbf{Assigned Type} \\
\midrule
(1) Two equal values in one row or column & Random \\
(2) Two pairs of equal values in rows or columns & Type I or II \\
(3) Three equal values & Random \\
(4) Four equal values & Random \\
\bottomrule
\end{tabular}

\label{tab:ties}

\end{table}

Note that assignments happen with uniform probability. For instance, $1/3$ for random assignment to a type and with a probability of $1/2$ in case 2. This is a deterministic way of dealing with ties based on local node behaviour.

\section{Results}
% \vspace{10pt}
\subsection{$CPE_G$ characterises the behaviour of the logistic map better than its classical discrete counterpart}
% \vspace{10pt}

The logistic map \cite{May76} is a popular time-series to show the performance of entropy measures in detecting non-linear changes. Given by
\begin{equation}
x_{n+1} = rx_n(1 - x_n)
\end{equation}
we used an initial value $x_{0}=0.65$ and incremented $r$ by steps of size $10^{-4}$ where $r \in [3.55,4.0$]. We created a time series for each value of $r$. Two underlying graphs were considered, a directed path and an undirected path on $N=3034$ vertices. Chaotic behaviour is well-studied to occur in the range $3.5699 \leq r \leq 4$. We computed both $PE_G$ and $CPE_G$ with various ordinal activation functions. 
\begin{figure}[h]
    \centering
    % Add your plot here
    \includegraphics[width=0.5\textwidth]{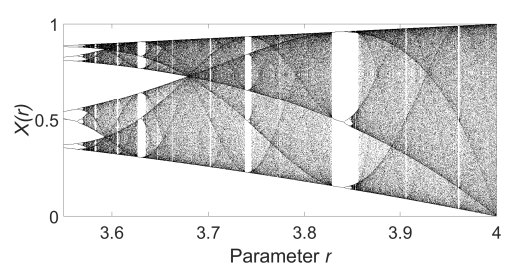}
    \caption{Bifurcation plot of the logistic map.}
    \label{fig:empty_plot}
\end{figure}

\begin{figure*}[t]
    \centering
    \begin{subfigure}{0.45\textwidth}
        \includegraphics[width=\linewidth]{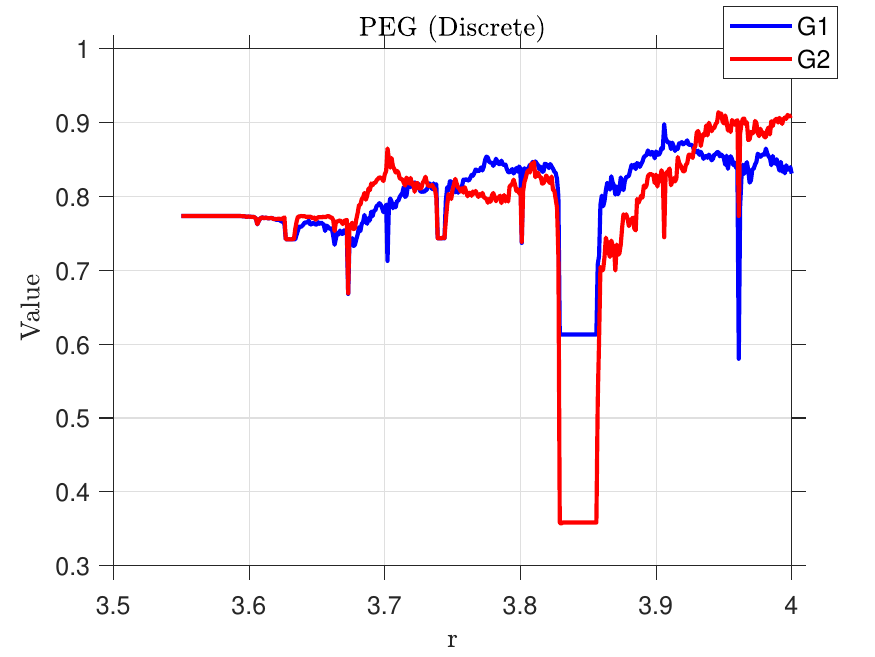}
        \caption{Discrete}
        \label{fig:sub1}
    \end{subfigure}
    \quad
    \begin{subfigure}{0.45\textwidth}
        \includegraphics[width=\linewidth]{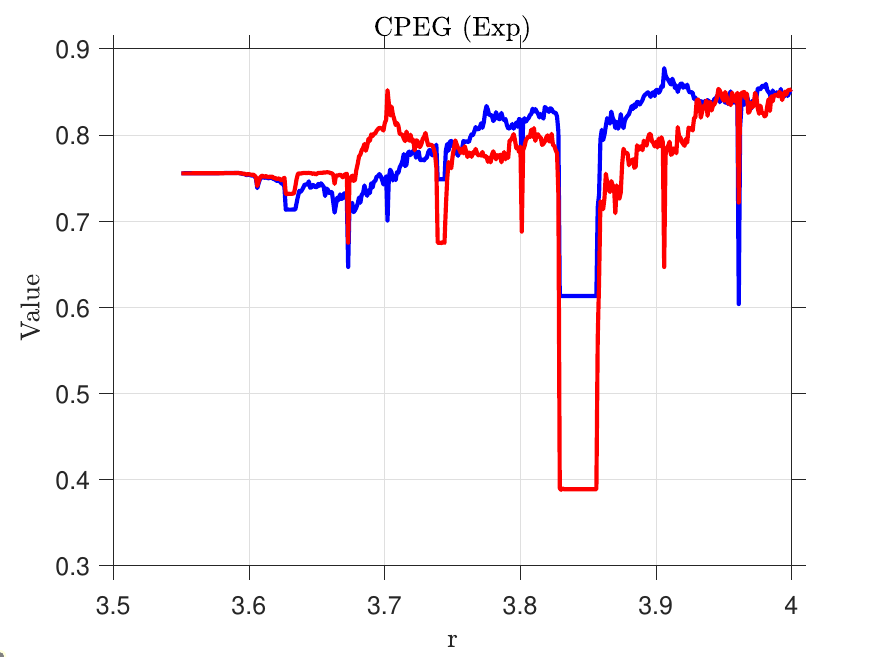}
        \caption{Exponential}
        \label{fig:sub2}
    \end{subfigure}
    \\
    \begin{subfigure}{0.45\textwidth}
        \includegraphics[width=\linewidth]{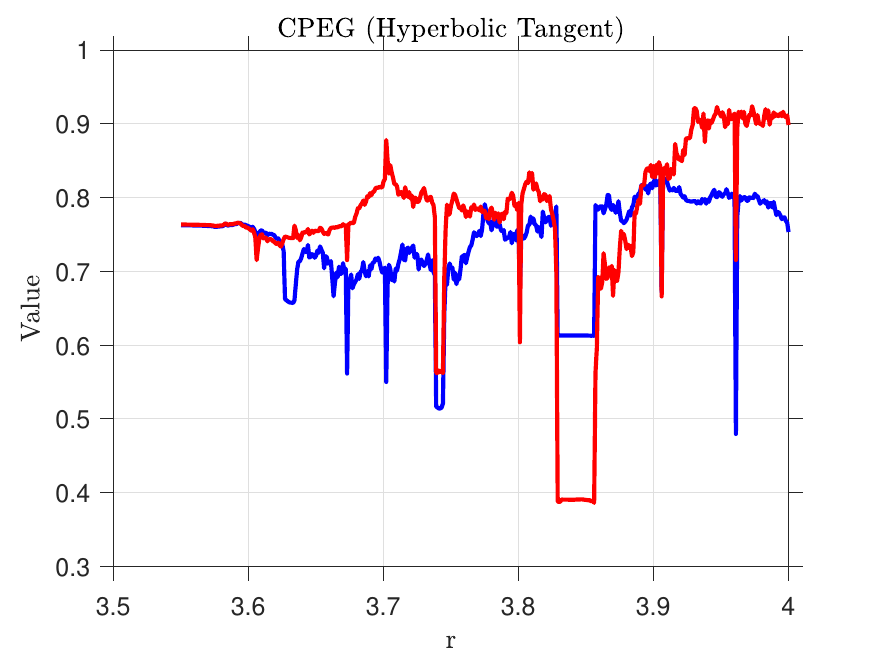}
        \caption{Hyperbolic Tangent}
        \label{fig:sub3}
    \end{subfigure}
    \quad
    \begin{subfigure}{0.45\textwidth}
        \includegraphics[width=\linewidth]{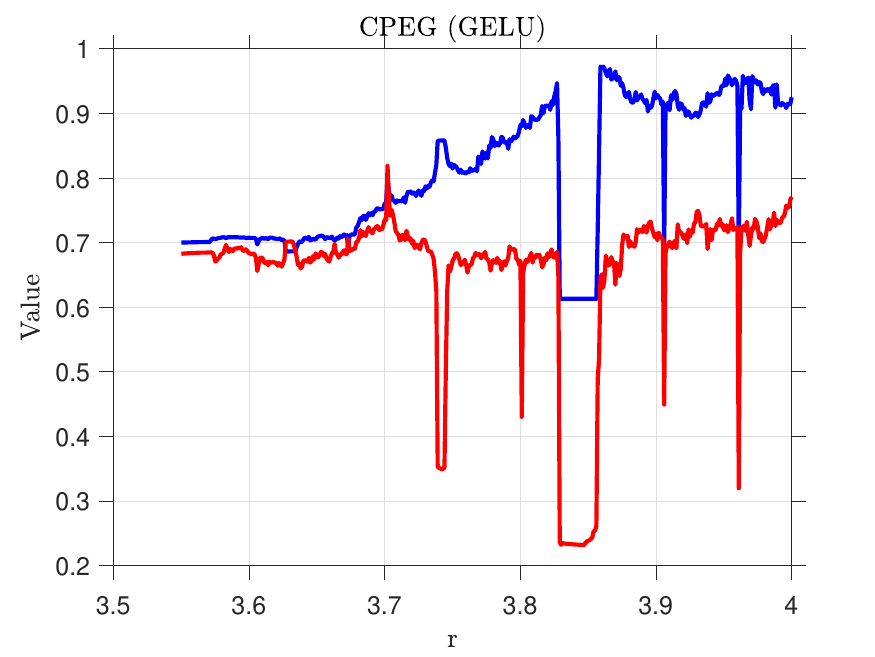}
        \caption{GELU}
        \label{fig:sub4}
    \end{subfigure}
    \caption{The Values of $PE_G$ and $CPE_G$ at different values of the parameter $r$ using various ordinal activation functions}
    \label{fig:log}
\end{figure*}

 Figure \ref{fig:empty_plot} shows the bifurcation plot of the logistic map. We observe mostly chaotic behaviour however there are islands of stability (white lines) with the most notable occurring between $3.8 \leq r \leq 3.9$. While both continuous and discrete versions of $PE_G$ can detect the large island of stability, the continuous version does seem to distinguish different change-points differently depending on the choice of ordinal activation function. Earlier islands of stability are also detected with sharper peaks with the continuous version, particularly in the undirected (G2, Red) case.

 In Figure \ref{fig:log} we notice that the exponential function is the one providing results most similar to the discrete version of $PE_G$. The hyperbolic tangent and GELU function seems to deviate more from the discrete case and can detect different islands of stability more clearly.
 This example illustrates how tuning the ordinal activation function can highlight various signal features of interest. While such tuning is beneficial, particularly in the context of deep learning, we must be cautious to avoid data dredging or overfitting. Importantly, even with adjustments to the activation function, the overall trends remain apparent in most cases.

We can see that $CPE_G$ can distinguish between periodic and chaotic activity better than the discrete, due to it robustly considering amplitude information during its computation. We also used the mean as our  statistic as that takes into account all the information in the window. We notice the plots are also significantly more detailed for the continuous case detecting subtle shifts in patterns in the chaotic signal.

Exploring this further, we assessed how much the probability distribution of $CPE_G$ deviated from $PE_G$ at different values of \( r \) in the logistic map for different ordinal activation functions. For this, we use the Kullback-Leibler Divergence \cite{kl} as our measure of deviation. The Kullback-Leibler Divergence (KL divergence) between two probability distributions \( P \) and \( Q \) is defined as:
\begin{equation}
    D_{\text{KL}}(P \| Q) = \sum_{i} P(i) \log\left(\frac{P(i)}{Q(i)}\right)
\end{equation}
and gives a measure of how one probability distribution diverges from a second.

\begin{figure}[h]
    \centering
    \includegraphics[width=0.5\textwidth]{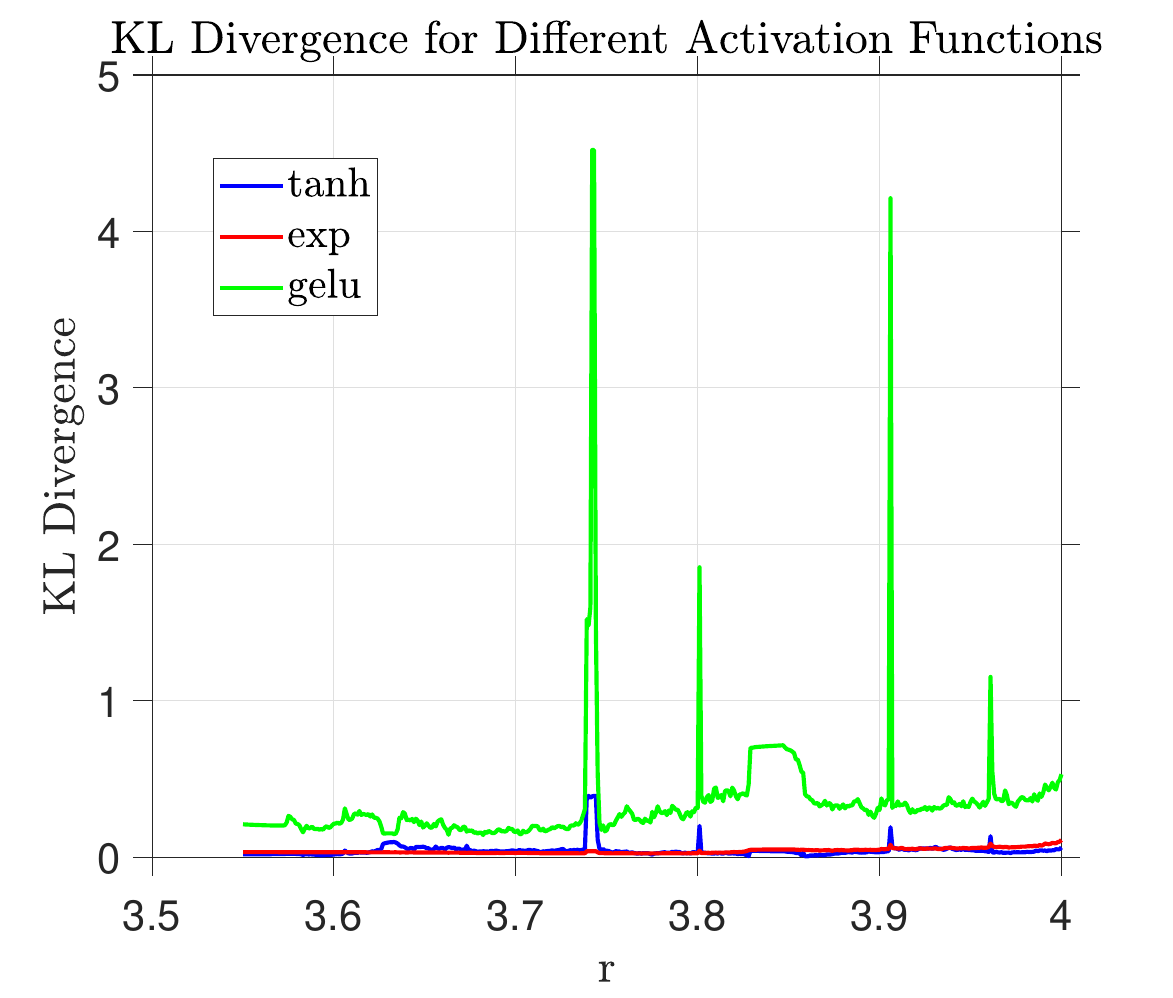}
    \caption{KL Divergence between $CPE_G$ and $PE_G$ in the Logistic Map with different Ordinal Activation Functions}
    \label{fig:kl}
\end{figure}

Inspecting Figure \ref{fig:kl}, all 3 ordinal activation functions remain stable for chaotic behaviour, i.e., their distributions are relatively similar to the discrete case. We note that functions with faster increasing gradients seem to be more stable relative to Discrete $PE_G$, while more slowly increasing gradients results in more deviation from the discrete case with large deviations particularly in the islands of stability.
In fact, the computation of the KL divergence between such functions and discrete $PE_G$ can actually detect islands of stability within chaotic behaviour. This implies that while the patterns that dominate the total frequency of patterns are similar to the patterns that contribute most to the overall power of the signal at islands of stability this is no longer the case.

It is useful to observe that in both the continuous $PE_G$ plot and the KL-Divergence plot that the global behaviour is very similar with subtle changes due to choice of ordinal activation function (the  statistic is kept constant as the mean) with most deviation in behavior occurring at the islands of stability. This is as expected and highlights how using different activation functions alters the behaviour at change-points allowing us to emphasize those that we deem relevant.

\subsection{Detecting non-linear spikes using Continuous Permutation Entropy}

We will now demonstrate the situation specific advantages where $CPE_G$ outperforms $PE_G$. 
We simulate an auto-regressive process of order 1 (AR(1)) with Gaussian non-linear spikes. The AR(1) process is defined by the equation \cite{AR1} :
\begin{equation}
X_{t+1} = \phi X_t + \varepsilon_{t+1}
\end{equation}
where $X_t$ represents the value of the process at time $t$, $\phi$ is the AR(1) parameter, and $\varepsilon_{t+1}$ is a Gaussian white noise with mean zero and standard deviation $\sigma$.

We set the parameters as follows: $\phi = 0.8$ and $\sigma = 1$. The simulation is conducted for $500$ time steps and $100$ realizations.

Additionally, we introduce Gaussian non-linear spikes into the process. The spike at time step $t$ centred at $t_{s}$ is represented by:
\begin{equation}
G(t)=A \cdot e^{-\frac{(t - t_s)^2}{2 W^2}}
\end{equation}
where $A$ is the amplitude of the spike, $t_{s}$ is the time step at which the spike occurs, and $W=2\sigma_{g}$ is the width of the spike with $\sigma_{g}$ the standard deviation of the spike. This spike time-series is added to the AR(1) process for 5 distinct time-steps $t_{s}=(50,100,200,250,450)$.

We use as our ordinal activation functions the exponential functionand the sigmoid function:
\begin{equation}
 \text{sigmoid}(x) = \frac{1}{1 + e^{-x}} 
\end{equation}
both with $T$ set as the mean of the embedding vector.

\begin{figure}[h]
    \centering
    \includegraphics[width=0.5\textwidth]{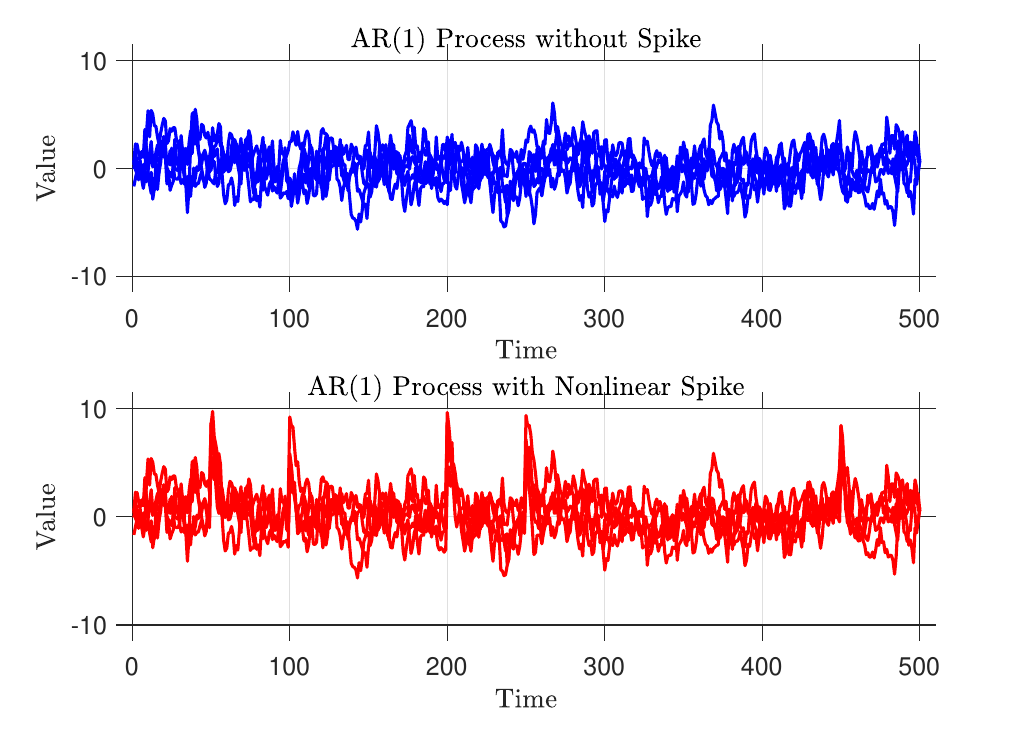}
    \caption{Comparison of AR(1) Process with and without Gaussian Spikes}
    \label{fig:ar1_comparison}
\end{figure}

Using an undirected path graph to map the time-series as a graph signal, we computed the $PE_G$ and contrast values for both ordinal activation functions and also for the discrete case for each of the 100 realizations of the signals with and without the Gaussian spike. Essentially we have two groups of 100 values (for each contrast and $(C)PE_G$) one with the spike and one without. We then performed the Wilcoxon rank-sum test for significance between the two groups.

\begin{table}[H]
\centering
\resizebox{\columnwidth}{!}{%
\begin{tabular}{|c|c|c|c|c|c|}
\hline
$PE_G$ & $\alpha$ & $\beta$ & $\gamma$ & $\tau$ & $\delta$ \\
\hline
\multicolumn{6}{|c|}{Exponential Function} \\
\hline
$8.8 \times 10^{-13}$ & $0.067$ & $3.2 \times 10^{-4}$ & $0.072$ & $0.067$ & $9.9 \times 10^{-8}$ \\
\hline
\multicolumn{6}{|c|}{Sigmoid} \\
\hline
$0.17$ & $0.36$ & $0.0018$ & $0.49$ & $0.36$ & $0.0018$ \\
\hline
\multicolumn{6}{|c|}{Discrete} \\
\hline
$0.60$ & $0.32$ & $0.040$ & $0.98$ & $0.32$ & $0.069$ \\
\hline
\end{tabular}%
}
\caption{p-values from the Wilcoxon rank-sum test for significance between parameter values across 100 realizations of graph signals with and without a Gaussian spike. The test was performed for $PE_G$ and the ordinal contrasts using the exponential and sigmoid ordinal activation functions for $CPE_G$, while discrete represents traditional $PE_G$.}

\label{tab:my-table}
\end{table}

Interestingly, as we inspect Table \ref{tab:my-table}, the sigmoid function, being a relatively slowly increasing function is not able to detect the non-linear spikes between signals with only the $\beta$ and $\delta$ contrasts detecting the difference in signals. The same is seen for the discrete case, with $PE_G$ failing to detect differences as well. The exponential functions mathematical properties allow it to detect these spikes very well, with $CPE_G$, $\beta$ and $\delta$ all detecting significant differences between signals. Given that the amplitude of the spikes is only a few values higher than the general behaviour of the signal, it seems that the exponential function is mapping these subtle differences to larger values than with functions with less-steep gradients, such as the natural log which would map the spike value similarly to general signal values. While we did see Discrete $PE_G$ perform well on directed graphs, many graph signals have underlying undirected graphs such as fMRI and EEG signals. We comprehensively demonstrate the benefit of $CPE_G$ in such a use case at detecting non-linear change-points in a signal.

\subsection{Real-world example: Analysis of weather data}
We analyse the real-life temperature readings of ground stations observed in Brittany, France, during January 2024 \cite{temp}. Here, the underlying graph is an un-directed weighted graph. The weights are determined using the Gaussian kernel of the Euclidean distance between vertices/ground station coordinates \cite{shuman}: 
\begin{equation}
\textbf{W}_{ij} = \begin{cases} 
\exp\left(-\frac{d(i,j)^2}{2\sigma^2}\right), & \text{if } d(i, j) \leq \sigma^2 \\
0, & \text{otherwise}
\end{cases}
\end{equation}

We first explored what extra information the (discrete) pattern contrasts could give about the graph signal on the underlying graph. We took the underlying graph structure   and computed the graph Laplacian eigenvectors and its eigenvalues. We then used each eigenvector as a graph signal and computed the graph permutation contrasts for each eigenvector. Note the underlying graph is completely unrelated to the graph signal itself and thus provides a general benchmark for the relationship of eigenvectors/eigenvalues of graphs with the ordinal contrasts. The results are displayed in Figure \ref{fig:tempeigen}.

\begin{figure}[h]
    \centering
    \includegraphics[width=0.5\textwidth]{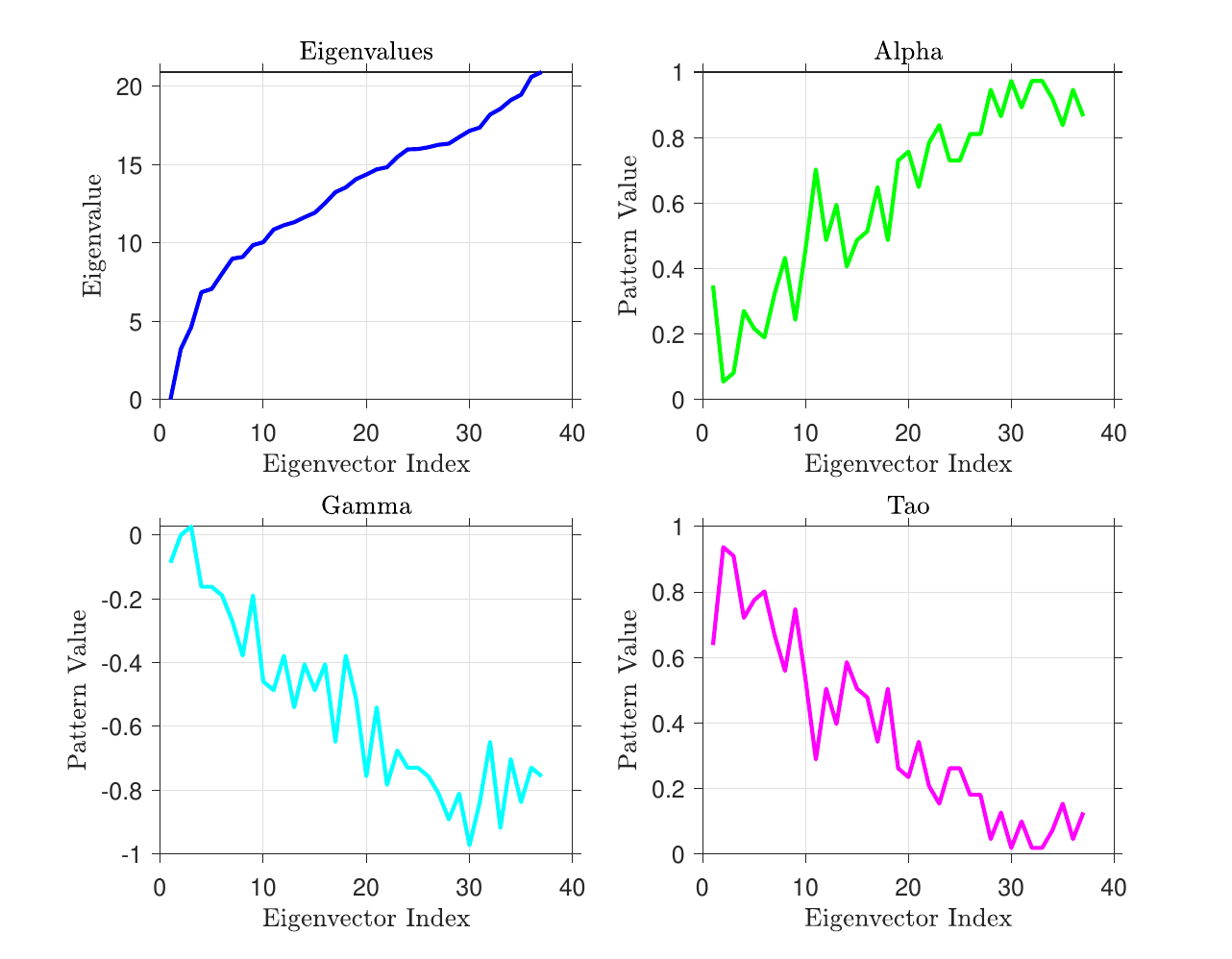}
    \caption{Eigenvalues of the Temperature Graph compared with the behaviour of the ordinal contrasts when we use each eigenvector as a graph signal}
    \label{fig:tempeigen}
\end{figure}

As hypothesized earlier in this work when we first introduced the ordinal contrasts (Methods Section C), the eigenvectors corresponding to larger eigenvalues seem to produce higher values of $\alpha$. The persistence $\tau$ is loosely related to the classical auto-correlation and can also be considered as a measure of smoothness as it differentiates between straight and broken patterns. It is also sometimes defined as $ 2/3- \alpha$ which is the negative of the turning rate plus constant 2/3 which is chosen such that $\tau$ is zero for white-noise (similarly for the other contrasts). In Figure \ref{fig:tempeigen}, we observe a clear negative correlation to the eigenvalues and intuitively to $\alpha$. While graph Laplacian eigenvalues are completely determined by the topology of the underlying graph, the $\alpha$  contrast takes both the actual signal values and underlying graph structure into consideration. This can provide a better representation of frequencies of graph signals.

Following this initial exploration, we assessed if we can detect changes in temperature at different time-periods. Spikes in temperature are an example of real life change-points. We noticed consistent seemingly parabolic behaviour in temperatures during the day hours compared to more linear behaviour in the night hours seen in Figure ~\ref{fig:temphours}.

\begin{figure}[h]
    \centering

    \includegraphics[width=0.5\textwidth]{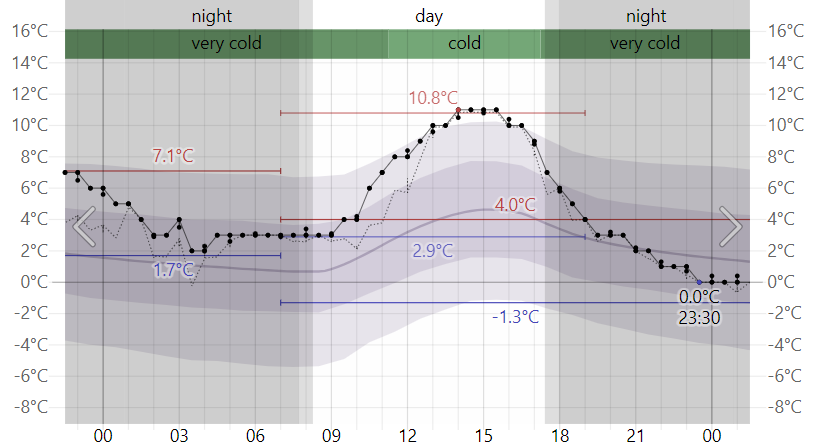}
    \caption{Temperature over hours on Friday, 17 January 2014 in Burgundy.}
    \label{fig:temphours}
\end{figure}

Following this, we extracted the graph signals for hours 12am-6am and for 9am-3pm the hours. We computed the $CPE_G$  and corresponding continuous ordinal values for all the hours for night and day during January, 2014. We computed the Wilcox on rank-sum test between groups and we used the hyperbolic tangent and exponential functions as our ordinal activation
functions and compared it to the discrete case.

\begin{table}[H]
\centering
\resizebox{\columnwidth}{!}{%
\begin{tabular}{|c|c|c|c|}
\hline
 & Discrete & Exponential & Hyperbolic Tangent \\
\hline
$PE_G$ & 0.20 & 0.0043 & 0.73 \\
\hline
$\alpha$ & 0.055 & 0.0033 & 0.47 \\
\hline
$\beta$ & 0.75 & 0.012 & 0.042 \\
\hline
$\gamma$ & 0.88 & 0.39 & 0.088 \\
\hline
$\tau$ & 0.055 & 0.0033 & 0.47 \\
\hline
$\delta$ & 0.00010 & 0.00090 & 0.00050 \\
\hline
\end{tabular}%
}
\caption{p-values for the Wilcoxon rank-sum test comparing night (12am-6am) and day (9am-3pm) groups in January 2014 using $PE_G$ (Discrete),  and $CPE_G$ with the exponential and hyperbolic tangent ordinal activation functions.}

\label{tab:temp}
\end{table}

 Table \ref{tab:temp} shows our results. The exponential ordinal activation function performs the best in terms of $PE_G$, picking up the more rapid increases in temperature very well. In fact, all of its contrasts, aside from $\gamma$, picked up the difference. In the discrete case, only the  $\rho$ contrast detected the spike, while in the hyperbolic tangent case the $\delta$ and $\beta$ contrasts picked up the spike. Additionally, it seems the $\beta$ contrasts performance improves with ordinal activation functions with larger gradients. Conversely, the $\delta$ contrast appears to detect the change regardless of the underlying ordinal activation function.

We have shown, in a real life undirected graph signal, how $CPE_G$ can be used to detect differences between day and night be detecting spikes in temperatures spatially distributed over power stations. We can see that the underlying activation function can influence $CPE_G$'s performance and can provide significant advantages to the discrete $PE_G$.

\begin{figure*}[htb!]
    \centering
    \includegraphics[width=\textwidth]{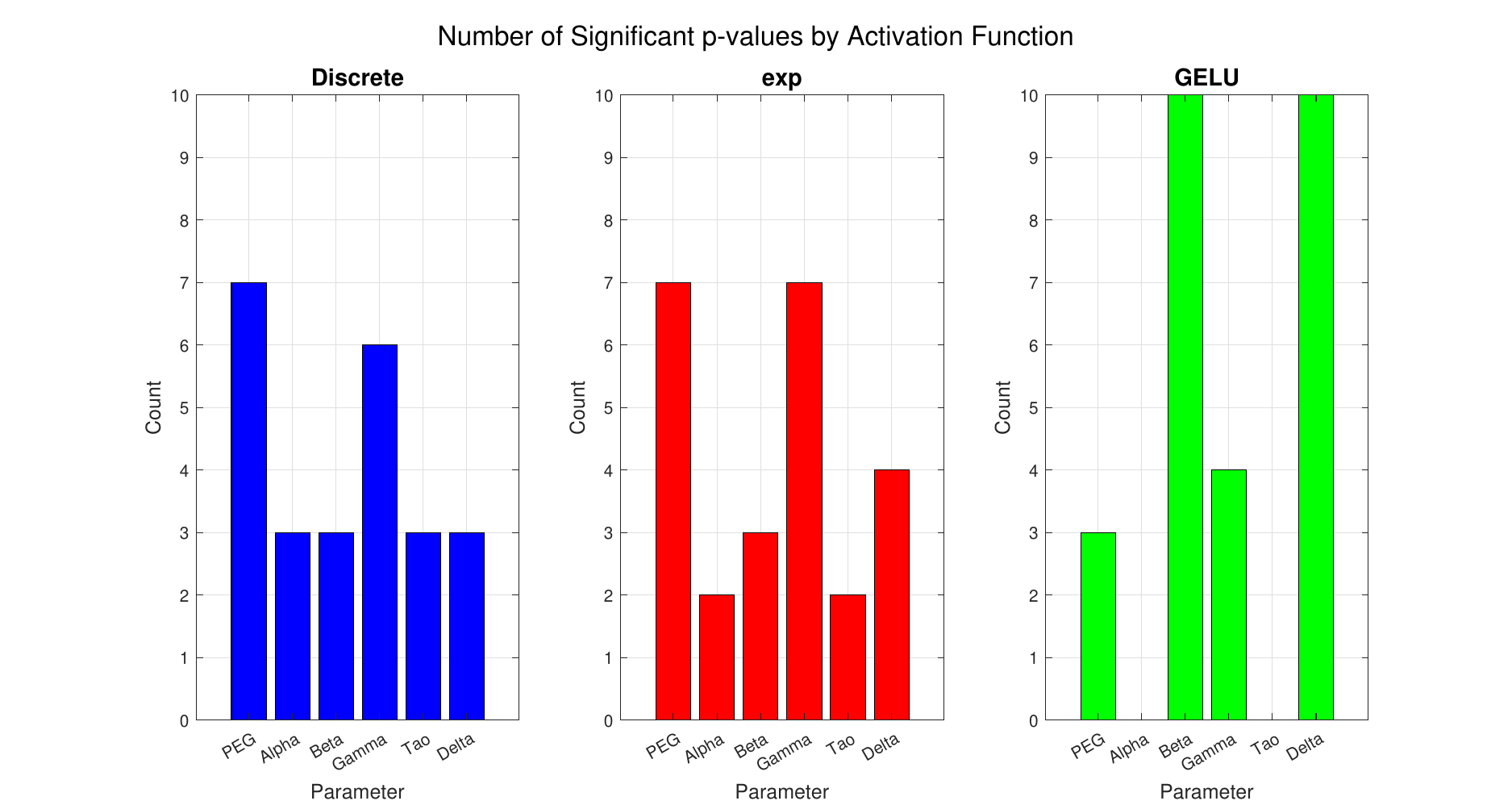}
    \caption{Number of significant p-values when testing Surrogate vs Normal data}
    \label{fig:sur}
\end{figure*}

\subsection{Heartbeat time-series}

The Fantasia database consists of 10 heartbeat time series: 5 correspond to young subjects (aged between 21 and 34 years) while 5 correspond to elderly subjects (aged between 68 and 85). It is a commonly used dataset to test entropy algorithms \cite{heartbeat,heartbeat1}. Each time-series consists of 4800 time-points. We split it up into 6 disjoint windows of 800 points and compute the $(C)PE_G$ and the corresponding contrasts  on the undirected path graph with 800 nodes.

Pirondoni et al.~\cite{Pirondini16} introduced a method to produce surrogate data specifically for graph signals using the Graph Fourier transform by taking the eigen-decomposition of the Graph Laplacian. We will apply this to create surrogate signals to test $CPE_G$'s ability to detect non-linearity in graph signals. 

Given the Laplacian of an undirected  graph is symmetric and real valued, it has a complete set of orthonormal eigenvectors: $V = [\textbf{v}_l]_{l=0,1,\ldots,N-1}$. The Graph Fourier Transform (GFT) coefficients of a graph signal $x$ are obtained by the projection $c = V^T \textbf{x}$. The inverse transform corresponds to $\textbf{x} = Vc$.

To generate the surrogate data \cite{Pirondini16}, we first take the GFT of the graph signal to obtain the coefficients $c$. The signs of $c$ are then randomly permuted. This preserves the amplitude of the GFT coefficients while destroying the non-linearity contained in the phase of the original signal. The inverse GFT is used to retrieve the surrogate signal $x$.

Then, we compute the $CPE_G$ on both the surrogate and normal graph signals form the heartbeat data. We perform the non-parametric Wilcox on Rank-sum test for significance between each patients signal and its surrogate. We compute the $CPE_G$  and the continuous contrasts with the exponential ordinal activation function for all 6 windows over all participants and evaluate if they can distinguish between the respective participant signal and its surrogate.  We repeat this for the discrete case.

 Figure \ref{fig:sur} illustrates the number of significant p-values identified by the Discrete PEG and its contrasts, as well as by $CPE_G$ and its contrasts, using the GELU and exponential ordinal activation functions. The results indicate a similar number of significant differences detected in both the discrete and exponential cases, with the $\gamma$  and $\delta$ contrasts outperforming the discrete contrasts, while the $\alpha$ and $\tau$ contrasts underperform. Notably, with the GELU function, the  $\beta$ and   $\delta$ contrasts consistently detect significant differences across all patients, whereas the $\alpha$  and $\tau$  contrasts detect none, and the gamma contrast detects fewer differences. This suggests that the GELU function with the $\beta$  and $\delta$  contrasts being the most effective in detecting non-linearity in graph signals, at least for the HRV data.

This deviation from the discrete case aligns with the previously presented K-L divergence plot. The analysis demonstrates that adjusting the ordinal activation function can be advantageous in specific contexts, essentially fine-tuning the gradient of the strictly monotonically increasing function to optimize the analysis. Gradient optimization methods such as Gradient Descent or Adam\cite{adam}, commonly used in deep learning during back propagation, could be potentially integrated and employed for this purpose.

%To explore this a bit further we decided to take inspiration from Bian et al. \cite{Bian12} we propose a slightly modified version of $CPE_G$ where all possible pattern length 3 ties are assigned to separate classes. This is due to the presence of repeated values in low resolution and Heart-rate variability data. This can significantly skew results if ties are not considered. Making this simple change we re-compute the FDR-corrected P-values.

%\begin{figure}[htbp]
 % \centering
  %\includegraphics[width=0.5\textwidth]{contpatternsmodified.png} 
  %\caption{$CPE_G$ Normal vs Surrogate signals}
  %\label{fig:example}
%\end{figure}

%Figure 5 shows that this simple change does infarct change results. Note we did not observe any change in the Logistic map when making this modification. While the interpretation of this is well beyond the scope of this study, it does show that not accounting for ties can result in skewed results.

\subsection{The Kylberg–Sintorn rotation dataset of textures: A Graph Pattern Contrast Perspective }

The Kylberg-Sintorn rotation dataset \cite{Kylberg16} contains  various textures from bulk solids and regular structures. We extract 10 images from each rotation ranging from 0-320 degrees in increments of 40 degrees. Each image has 122x122 pixels with gray values normalized with a mean of 127 and a standard deviation of 40. We used the hardware rotation, which is performed by turning the camera. We used the images of rice, lentils and fabric 5 from the dataset(Figure \ref{fig:texture}). Lentils and rice are isotropic images while fabric 5 has horizontal and vertical lines that may make the image less rotation-invariant. For each picture, we used our $PE_G$ image algorithm with $m=4$, and a regular  grid graph and embedding delay $L=1$. Note that we could have used a larger embedding delay but since the image is quite small we employed the minimum value. We plotted the results on the $\theta-\kappa$ plane, which can be interpreted similar to  the classical entropy ($\theta$)-complexity ($\kappa$) plane.

\begin{figure}[htb!]
    \centering
    
    \begin{subfigure}{0.4\textwidth} 
        \centering
        \includegraphics[width=\linewidth]{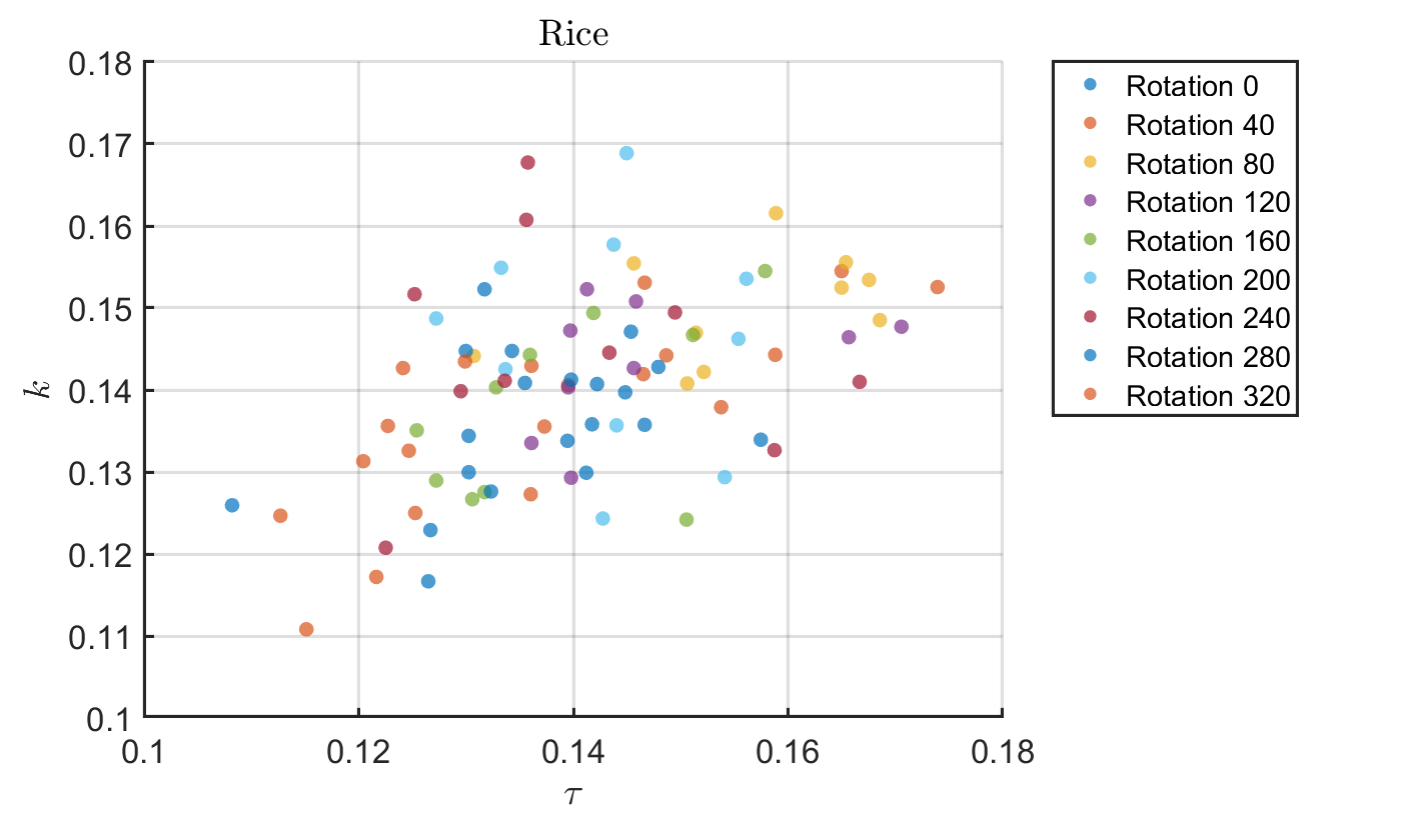}
        \caption{Rice}
        \label{fig:sub1}
    \end{subfigure}
    \hfill
    \begin{subfigure}{0.4\textwidth} 
        \centering
        \includegraphics[width=\linewidth]{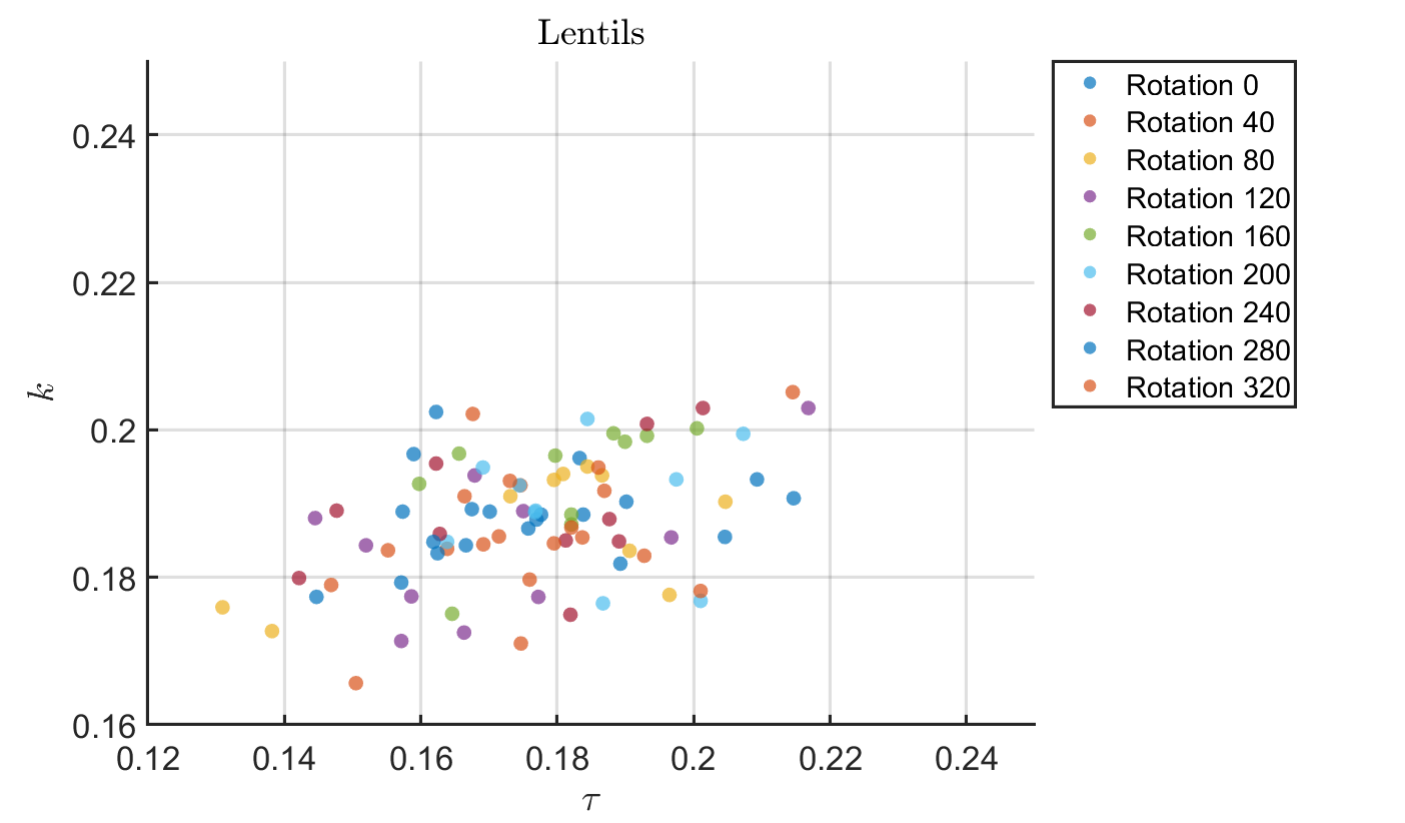}
        \caption{Lentils}
        \label{fig:sub2}
    \end{subfigure}
    
    %\vspace{0.5cm} 
    
    \begin{subfigure}{0.4\textwidth} 
        \centering
        \includegraphics[width=\linewidth]{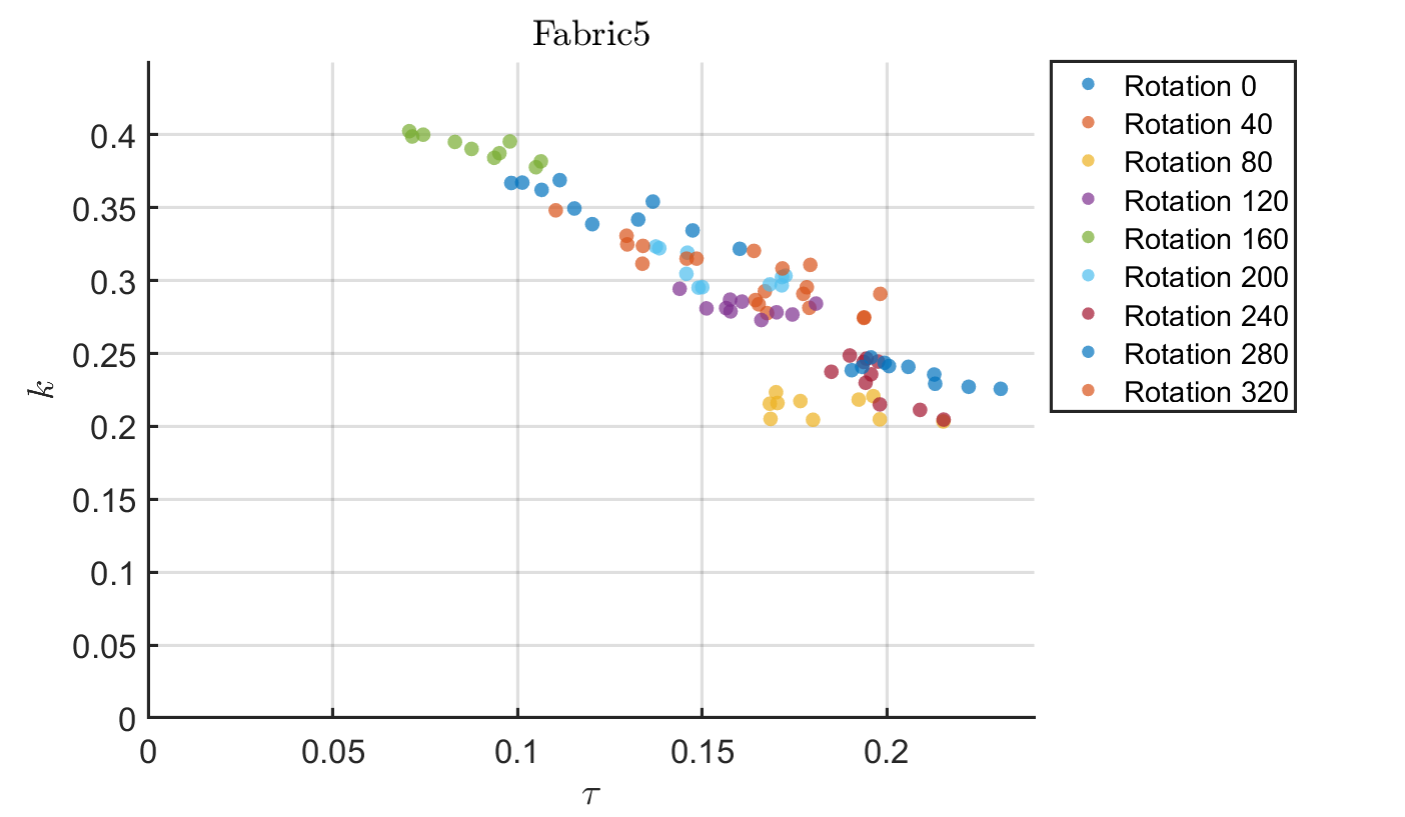}
        \caption{Fabric5}
        \label{fig:sub3}
    \end{subfigure}
    
    \caption{The Textures on the $\theta-\kappa$ plane for various rotation values.}
    \label{fig:vertical}
\end{figure}

 It is interesting to note (in Figure \ref{fig:vertical}) that for the isotropic Rice and lentils textures, all the images have similar $\theta$ and $\kappa$ values. The range of the values for lentils are $\kappa \in [0.16,0.18]$ while for theta $\in [0.13,0.22]$. For the rice textures, there are $\kappa \in [0.11,0.17]$ while for $\theta \in [0.11,0.175]$. As we can see, these values are similar in both textures and encompass a relatively small range of values. Visually, we observe these results in the detection of the texture regardless of rotation, i.e., rotation invariant detection. The distinguishing feature of the two textures is the value of $\theta$. Recall that $\theta$ is similar to persistence in the length-3 pattern case, where we would also expect it to be positive in practice with the maximum ($2/3$) representing a monotonic function and minimum ($-1/3$) a highly oscillating function. It is intuitive that lentils and rice are close to the middle of this range as they are somewhat monotonic depending on the light in the picture (i.e., lighter gray scale values have higher pixel values). However, they are also oscillating slightly between light and dark gray values. $\kappa$ describes the existence of branching or spiral like structures in the image. This is not readily present in the rice or lentil textures resulting in the value remaining low. The main distinction occurs in fabric5. Whereas with rice and lentils, all the rotations of the images are part of the same cluster. It is evident that this is not the case with Fabric5, which is due to the image not being isotropic. We can see for example rotation 180 and 0 are close to each other as expected when vertically flipping the image. However, there are big gaps between the other 40 degree rotations, i.e., rotation 80 is in  a clearly distinct cluster from rotation 180. We note higher values for $\kappa \in [02,0.41]$ as expected due to the clear branching structure in the image. Furthermore, $\theta \in[0.06,0.23]$ has a distinctly larger range due to the image not being isotropic. Overall, our $PE_G$ Image patterns translate well to the graph domain with the benefit being that we also get a pixel by pixel granularity and can change the underlying graph to emphasize certain properties of the image.

\begin{figure}[H]
    \centering
    
    \begin{subfigure}{0.15\textwidth} 
        \centering
        \includegraphics[width=\linewidth]{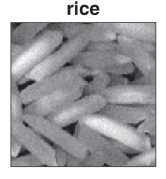}
       
        \label{fig:sub1}
    \end{subfigure}
    \hfill
    \begin{subfigure}{0.15\textwidth} 
        \centering
        \includegraphics[width=\linewidth]{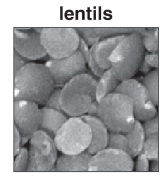}
       
        \label{fig:sub2}
    \end{subfigure}
    
    \vspace{0.5cm} 
    
    \begin{subfigure}{0.15\textwidth}
        \centering
        \includegraphics[width=\linewidth]{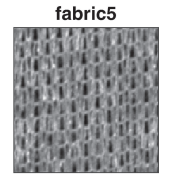}
       
        \label{fig:sub3}
    \end{subfigure}
    
    \caption{Textures from the  Kylberg–Sintorn rotation dataset.}
    \label{fig:texture}
\end{figure}

\subsection{Analyzing Fractal Surfaces using Graph Image Contrasts}

Fractal surfaces can be produced using a stochastic algorithm called the Random midpoint displacement algorithm \cite {midpoint}. In essence, the algorithm calculates the midpoints of existing elements and adds random number from a given distribution to them. The algorithm allows us to specify the size of the fractal as an image and also allows a roughness value between 0 and 1 that determines the roughness of the generated fractal surface (with 0 being the least rough and 1 being the most). We generated 100 surfaces each for roughness values between 0 and 1. Each fractal surface was of size $33 \times 33$. We used, once again, a regular  grid graph as our underlying graph structure. We also varied the delay  component $L$ of our algorithm from 1 to 6 to see if that influenced results. We plotted each value in the $ \theta- \kappa$ plane as before.

\begin{figure*}[htb!] 
    \centering
    \includegraphics[width=\textwidth]{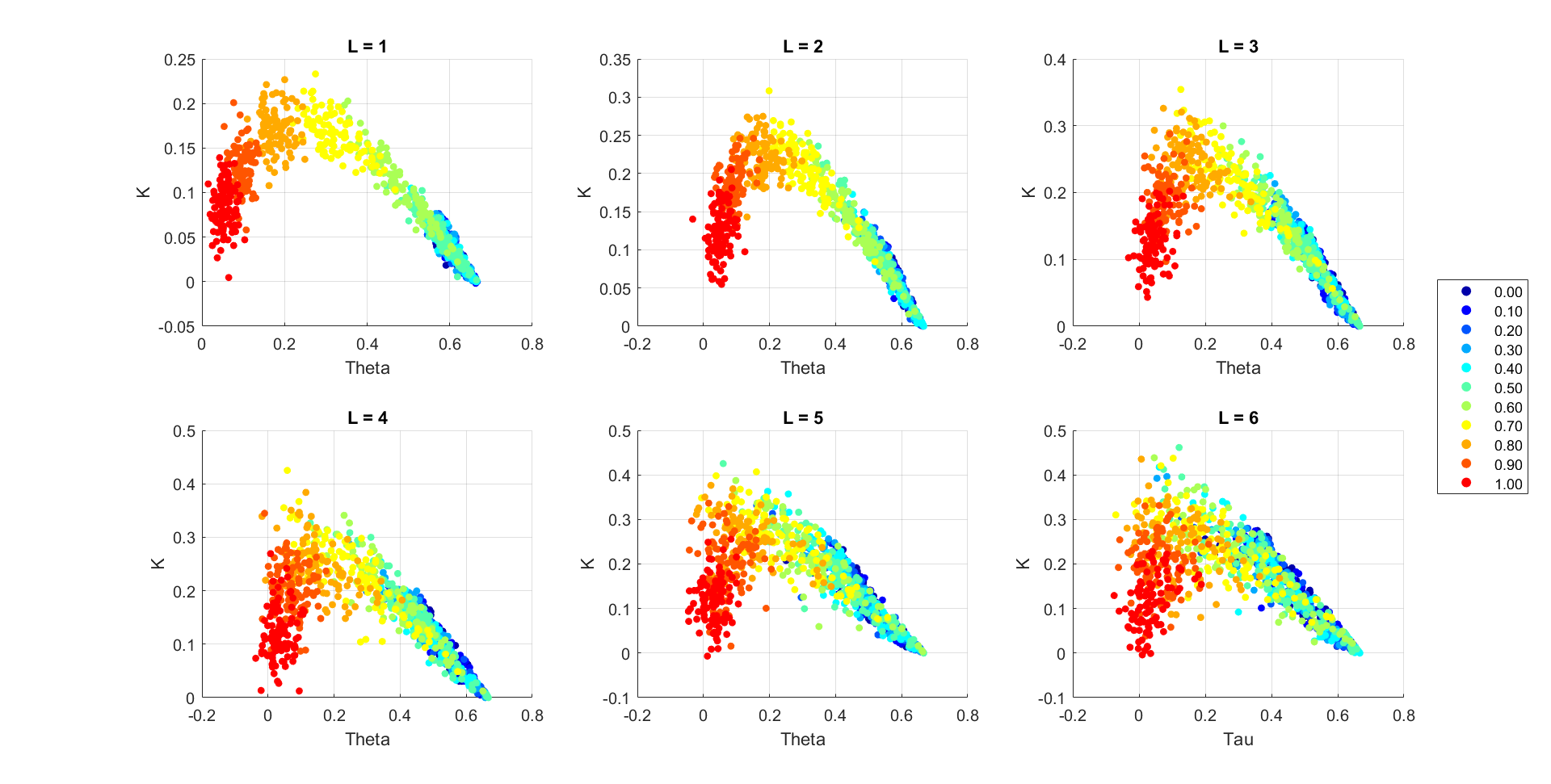}
    \caption{Random Midpoint Displacement Fractal (RMDF) in the $\theta- \kappa$ plane for different values of roughness and $L$.}
    \label{fig:rmdf}
\end{figure*}

Figure \ref{fig:rmdf} shows the Random Midpoint Displacement Fractal (RMDF) in the $\theta - \kappa$ plane and \ref{fig:fsp} shows the surface plots of the fractals at different values of R.
We varied the delay $L$ and the roughness parameter. For all delays the plot displays a structure resembling a convex negative parabola. There are clusters for different values of roughness and we notice that as $L$ increases, the clusters seem to intersect more. However, the main parabolic structure is still maintained regardless of the value of $L$. We can see when the roughness is near 0, $\theta$ approaches its maximum value of $2/3$. On the other hand, $\kappa$ is lower for roughness values near 0. Following our interpretation, this suggests that there are less branching or spiral like structures in the fractal surface. The maximum of the parabola occurs in the roughness range $[0.7,0.8]$. Here $\kappa$ is at a maximum. If we consider $\kappa$ as complexity, we can see that the complexity is at a peak where there is a majority of rough structures but still some smoother areas in the fractal surface. $\kappa$ drops as we approach a roughness of 1 and $\theta$ drops monotonically as the roughness increases. This behaviour is expected, as when roughness=1 we have largely irregular surfaces which is akin to strong, frequent oscillations or increases in the relative number of turning points compared to monotonic patterns.

\begin{figure*}[htb!]
    \centering
    \includegraphics[width=1\linewidth]{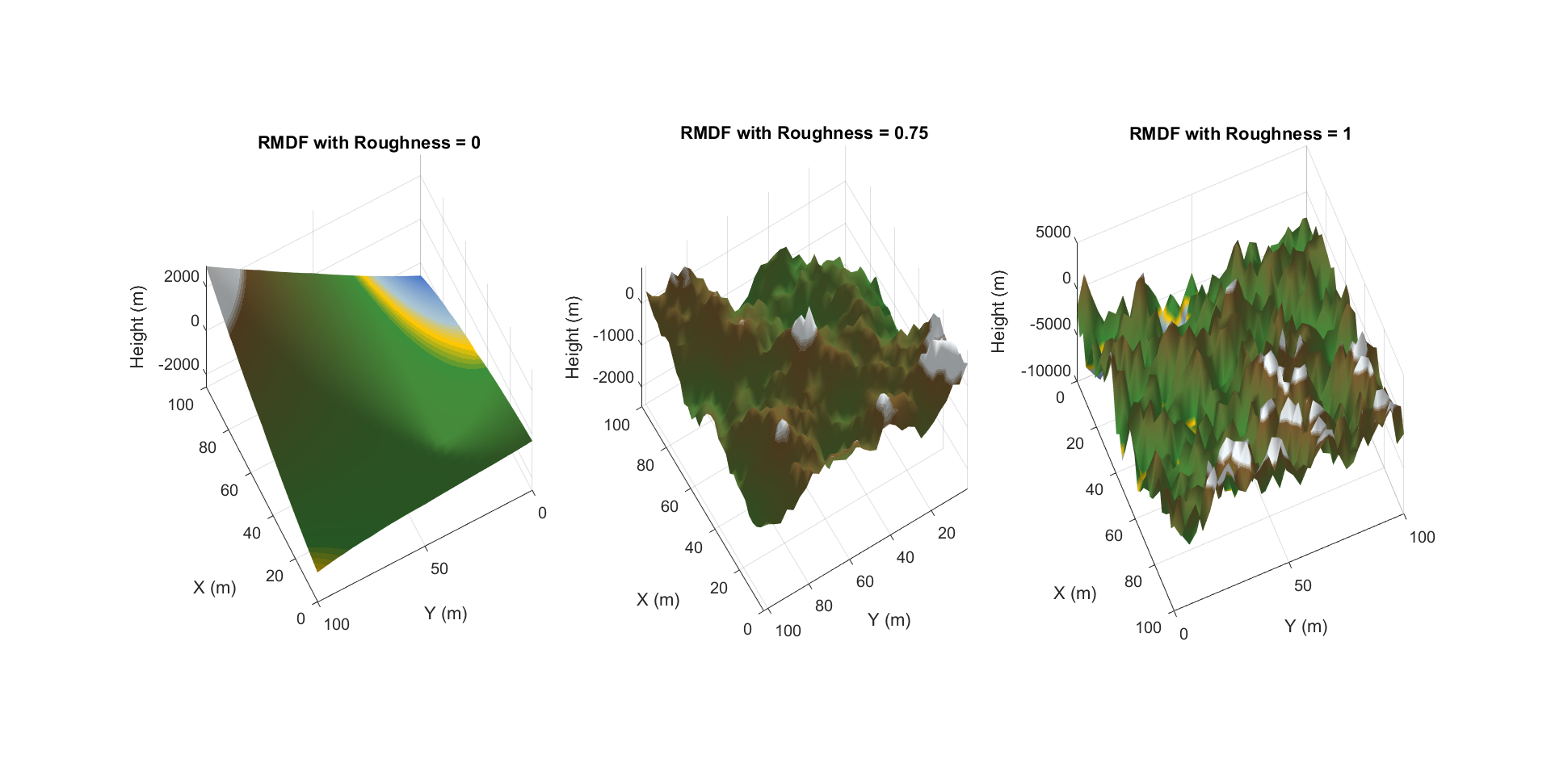} 
    \caption{Fractal Surface Plots at different values of Roughness (R).}
    \label{fig:fsp}
\end{figure*}

 Figure \ref{fig:ec} shows a rough morphology  of  the Jensen–Shannon complexity vs normalized entropy plane introduced by Rosso et al.~\cite {rosso}. We can see nearly identical behaviour of the plane with our image contrasts, this confirms that our graph based approach to the ordinal contrasts introduced by Zanin can distinguish between stochastic or chaotic behaviour and less complex behaviour in the analysis of fractal surfaces. At extremes of values of $R$ we can note low values of $\kappa$ correspond with either periodic behaviour of the fractal surface near $R=1$ or `white-noise' like behaviour near $R=0$. Chaotic and Stochastic behaviour, being notoriously hard to distinguish between is seen near $R=0.75$, we can see at values of $R=0.7,0.8$ the points are overlapping the chaotic and stochastic areas of the entropy-complexity plot (Figure \ref{fig:ec}). At $L=4$ we can see some distinguishing behaviour with $R=0.7$ showing more chaotic behaviour and $R=0.8$ showing more stochastic behaviour.

\begin{figure}[H]
    \centering
    \includegraphics[width=1\linewidth]{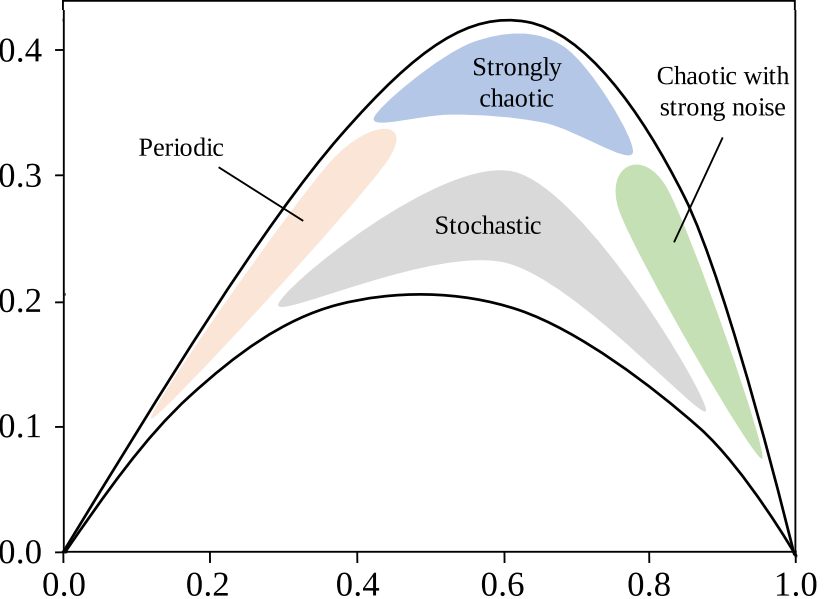} 
    \caption{Entropy-Complexity plane from Figure 5 of \cite {minwoo}.}
    \label{fig:ec}
\end{figure}

\subsection{$CPE_G$ as a method to characterize regional brain changes in Mild Cognitive Impairment}

Alzheimer's disease (AD) is the leading cause of dementia characterized by a substantial and progressive decline in cognition~\cite{alz}. It affects over 50 million people world wide, expected to triple by 2050, and places an immense burden on families and healthcare services. AD often progresses from stages of amnestic (memory-related) Mild Cognitive Impairment (aMCI), where the predominant symptom is memory loss greater than what is expected in natural aging~\cite{mci}. We explore the progression of aMCI using $CPE_G$, following a similar methodology presented in~\cite{john24} which explored $PE_G$ for the same data.

Participants were selected from this longitudinal study~\cite{mario}, which included patients with early Mild Cognitive Impairment (eMCI), MCI, and Mild Cognitive Impairment converters (MCIc). MCIc were patients who presented with MCI which converted to AD after a 2-year follow up. From this study, 8 healthy controls (Age: $76.50\pm5.21$, Sex: 2M; 6F), 7 eMCI (Age: $76.86\pm6.41$, Sex: 4M; 3F), 10 MCI (Age: $72.30\pm5.64$, Sex: 5M; 5F), and 6 MCIc subjects (Age: $76.33\pm5.09$, Sex: 4M; 2F) were selected. These subjects underwent functional magnetic resonance imaging (fMRI) scanning during which they performed a visual short-term memory binding task (VSTMBT) and diffusion tensor imaging (DTI). 

The VSTMBT is a task sensitive to early changes in AD targeting the retention of coloured shapes in memory~\cite{brain}. Participants are presented a set of non-nameable coloured shapes for 2 seconds (encoding phase), must memorize these over a variable window (2-8 seconds), and then are presented a new set of coloured shapes for 4 seconds. They must then determine if the new set of coloured shapes are the same or different, before the task is repeated after an inter-trial interval. In this study, as in~\cite{john24}, we explore the encoding phase of the task to assess changes in the formation of memory over the stages of MCI.

fMRI and DTI pre-processing are detailed in~\cite{john24}. Of note, regions of the brain for both fMRI and DTI were segmented with a modified version of the Desikan-Killiany atlas described in~\cite{nimage}, along with the brain-stem. This resulted in 85 regions, which were used to construct our brain signals and networks.

We explored the mean fMRI signal for the encoding phase of our task at each of our 85 brain regions. The underlying topology, or graph, is determined by our DTI networks (white-matter streamline density between brain regions). Prior to applying $CPE_G$, we use $z$-score normalization across the nodes to account for large values in the fMRI signal interacting poorly with our activation function.

Using the fMRI signals and DTI graph, we calculate $CPE_G$, as in step 6 of Definition~\ref{def:$CPE_G$}, to obtain a continuous value and the classical $PE_G$ pattern at each node for $m=3$. This is computed with the mean as the  statistic ($T$) and the exponential activation function. 

    For each comparison of control and disease, we calculate the relative, activation function weighted, frequency for each node across both groups using their ordinal pattern and continuous value. This is used to find per node statistical differences in control and disease using a standard t-test. Our hypothesis was that changes in amplitude, captured by the nodal continuous values of $CPE_G$, would highlight additional changes due to disease not captured by traditional $PE_G$'s permutation patterns alone. In addition, we asses the stability of the resulting t-test $p$-value by randomly permuting control and disease groups 1000 times, and observing the cases where our original $p$-value is smaller than the permuted case ($p'$). Given the limited sample size of our cohort, we chose a conservative approach and report regions where $p,p'\leq0.5$ as seen in Table~\ref{tab:MCI}.

\begin{table}[htb!]
\begin{tabular}{llll}
\textbf{Control vs.} & \textbf{ROIs}              & \textbf{$p$-value} & \textbf{$p<p'$} \\ \hline
\textbf{eMCI}        & Right-lateralorbitofrontal & 0.016              & 0               \\
                     & \textbf{Right-parstriangularis }    & 0.019              & 0.009           \\
                     & \textbf{Right-caudate}              & 0.022              & 0.006           \\
                     & \textbf{Left-thalamus}              & 0.033              & 0.005           \\ \hline
\textbf{MCI}         & Right-entorhinal           & 0.012              & 0.002           \\
                     & Right-lateralorbitofrontal & 0.018              & 0.002           \\
                     & Right-parahippocampal      & 0.027              & 0.015           \\
                     & \textbf{Right-postcentral }         & 0.036              & 0.005           \\ \hline
\textbf{MCIc}        & Right-lateralorbitofrontal & 0.002              & 0               \\
                     & Left-hippocampus           & 0.012              & 0.003           \\
                     & Right-paracentral          & 0.014              & 0               \\
                     & \textbf{Left-inferiortemporal}      & 0.015              & 0.005           \\
                     & Left-medialorbitofrontal   & 0.031              & 0.004           \\
                     & \textbf{Left-lingual}               & 0.031              & 0.003           \\
                     & \textbf{Right-isthmuscingulate}     & 0.032              & 0               \\
                     & \textbf{Left-thalamus}              & 0.033              & 0               \\
                     & Left-caudalmiddlefrontal   & 0.037              & 0.002          
\end{tabular}
\caption{Statistical tests comparing controls to the various stages of MCI. Regions in bold are new regions identified using relative frequency calculated from $CPE_G$. Non-bold regions were additionally identified by $CPE_G$, while also appearing in the preceding $PE_G$ study~\cite{john24}.}
\label{tab:MCI}
\end{table}

We find that changes in the continuous ordinal pattern distribution, calculated with $CPE_G$, captured not only the same regions across the disease stages as previously found in $PE_G$~\cite{john24} but also additional regions in all three stages (3 in eMCI, 1 in MCI, and 4 in MCIc), shown in bold in Table~\ref{tab:MCI}. In~\cite{john24}, $PE_G$ identified regions consistent with the anatomical trajectory of the AD continuum. $CPE_G$ captures these same regions along with additional changes in the basal ganglia (thalamus, caudate), inferior frontal and temporal gyruses (pars triangularis, inferior temporal), lateral parietal lobe (post central), and cingulate gyrus (isthmus cingulate). 

Looking at Table~\ref{tab:MCI} across the stages of MCI, three distinct neural pathways emerge related to functions of the VSTMBT and AD related neurocognitive deficits.

In eMCI, we observe changes in $CPE_G$ for regions along the sensory to subcortical to cortical pathway, involving visual inputs, the basal ganglia, thalamus, and motor language areas of the frontal lobes. The exact function of this pathway remains unclear, but it may be linked to strategy development in response to the early stages of disease, and possibly indicating a compensatory process \cite{parra13}. This pathway, associated with a broad range of cognitive and sensorimotor functions, including language \cite{klos23}, could signal the potential engagement of compensatory strategies to encode materials verbally and attempt to rehearse them, although such strategies are unlikely to support task performance effectively. These strategies may become less apparent or unavailable as the disease progresses to more advanced stages (see also \cite{valdes20,parra14} for frontal and basal ganglia involvement, respectively).

Patients with MCI exhibited a clear change in regions involved in the memory pathway, aligning with current notions of medial temporal lobe involvement in context-free memory tasks (VSTMBT) and in early AD progression~\cite{didic11,parra22}. This network supports inputs from the visual ventral stream into the medial temporal lobe (anterior to the posterior network), where object identity (anterior MTL) \cite{star10,zach14} is processed and fed into the associative network (posterior MTL) to form episodic memories \cite{mayes07}.

Patients with more advanced MCI (MCIc) exhibited changes in $CPE_G$ in both of the previously identified pathways in the earlier stages of disease. Additionally, for MCIc, we identified regions that connect the previously discussed pathways to neocortical areas of the brain. These regions may support the propagation and spread of neuropathology and associated neurocognitive deficits. This finding is particularly interesting and potentially novel, highlighting the strength of the amplitude-aware CPEG approach. This approach reveals regions connecting the previously mentioned areas to other neocortical areas. The non-linear nature of CPEG, with its ability to detect spikes, may increase sensitivity to detect the propagation of neuropathology, which is expected in MCI patients who are likely to develop dementia \cite{insel20}.

To further illustrate these results, we generate brain graphs in Figure\ref{fig:braingraphs}. Labeled nodes are those present in Table \ref{tab:MCI}. Nodes in orange represent a change in dominant pattern, defined as a difference in pattern that exists in at least half of the subjects. Blue nodes had no change in dominant pattern, while black nodes exhibited no definitive change. One interesting result to note is that of the right isthmus cingulate (rIST in Figure\ref{fig:braingraphs}c). The rIST exhibited no change in the discrete dominant pattern, but had a significant change in weighted frequency due to the continuous value determined by the activation function. This highlights the utility of $CPE_G$. With no change in dominant pattern, $PE_G$ would not identify that region as statistically significant. However, due to $CPE_G$'s increased spike detection capabilities due to the exponential function combined with its amplitude-awareness, a consistent and significant (subtle) change in amplitude in controls vs.~MCIc was identified, despite no change in the discrete dominant pattern. \\ 

\begin{figure}[htb!]
    \centering
    
    \begin{subfigure}{0.2\textwidth} 
        \centering
        \includegraphics[width=\linewidth]{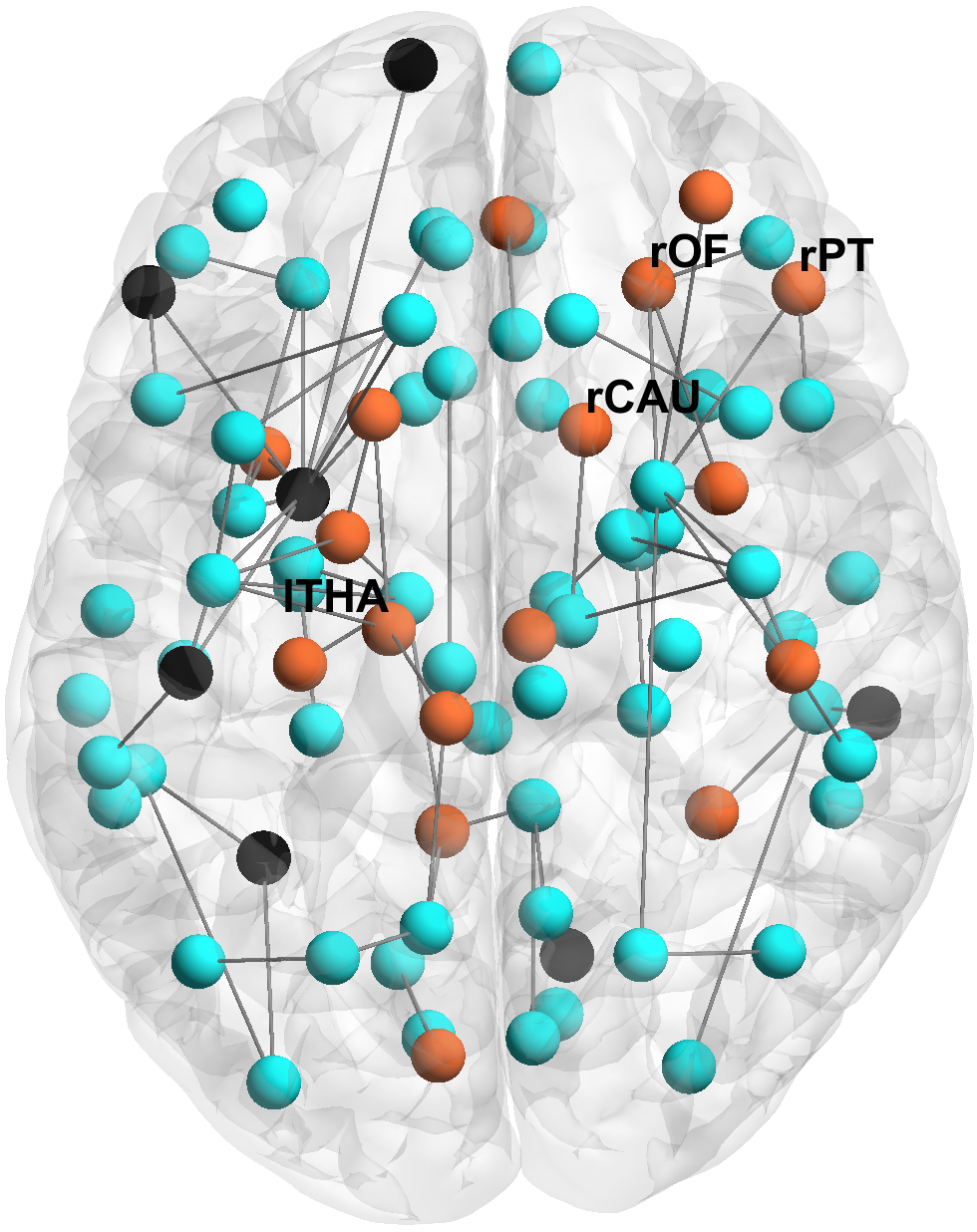}
 \centerline{\small (a) control vs.~eMCI}
    \end{subfigure}
    \hfill
    \begin{subfigure}{0.2\textwidth} 
        \centering
        \includegraphics[width=\linewidth]{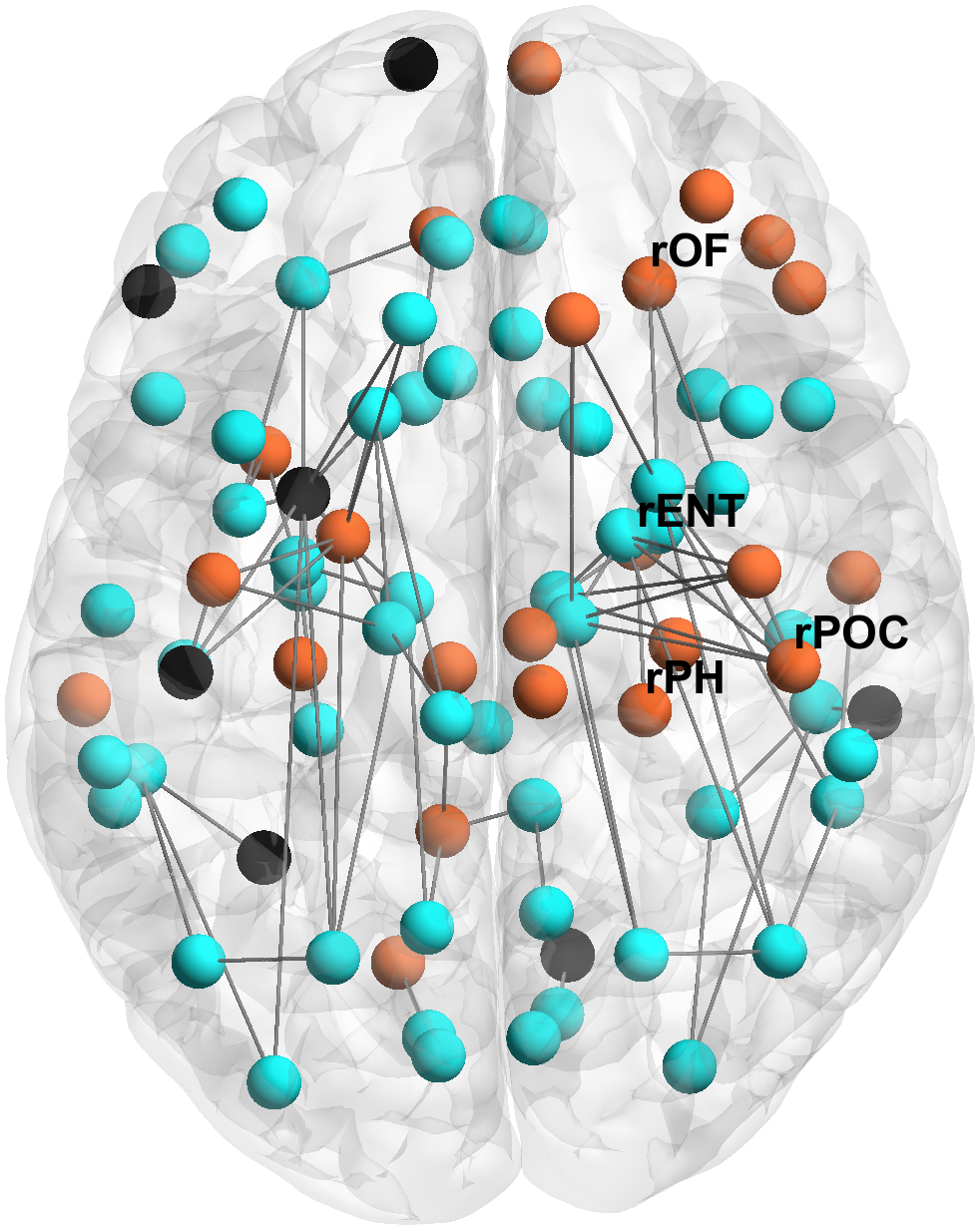}
 \centerline{\small (b) control vs.~MCI}
    \end{subfigure}
    
%    \vspace{0.5cm} 
    
    \begin{subfigure}{0.2\textwidth}
        \centering
        \includegraphics[width=\linewidth]{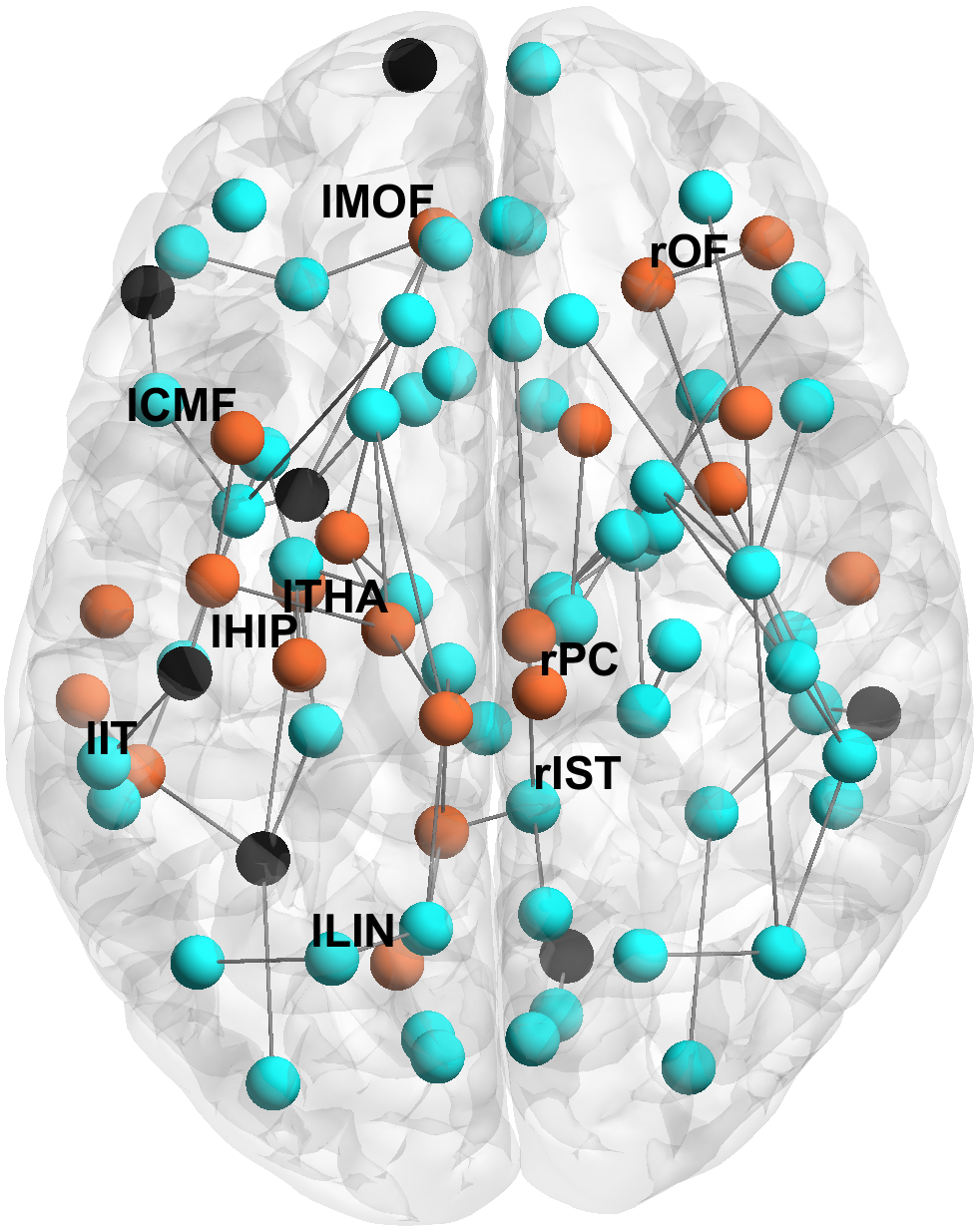}
\centerline{\small (c) control vs.~MCIc}
    \end{subfigure}
    
    \caption{Visualization of the changes in pattern between control and disease. Here, 2\% of the underlying graph's (DTI) edges are visualized for clarity. Brain graphs were generated with BrainNet viewer~\cite{bviewer}.}
    \label{fig:braingraphs}
\end{figure}
\newpage

\section{Discussion}

In this paper, we have comprehensively extended various cutting-edge, recent advancements in non-linear dynamic analysis, in particular in the analysis of Permutation Entropy and ordinal patterns. We have formally defined pattern contrasts for Graph Signals and have shown that they provide additional information compared to $PE_G$ alone. Furthermore, these contrasts are effective at detecting islands of stability in the presence of non-linear activity, showing its ability to recognize and discriminate non-linear activity from more stable behaviour. We also showed that the $\alpha$ and $\tau$ pattern contrasts are related to the topology of the graph and its Laplacian eigenvectors/eigenvalues.

We introduced an image specific extension of pattern contrasts to the graph domain by extending a quantization method to the graph domain to patterns of length 4, showing that this is akin to an analysis in the entropy-complexity plane. We also showed that our method is rotation-invariant for isotropic images and can distinguish between different types of images. Note that for all images we used a regular grid graph. Extending this further, we analyzed fractal surfaces generated by the Random Midpoint Displacement Algorithm. Our introduced image contrasts in the graph domain could distinguish different levels of roughness and showed parabolic like behaviour for incremental levels of roughness in the $\theta-\kappa$ plane. We showed that our interpretation of $\theta$ and $\kappa$ seemed to hold with our extension to the graph domain, and that no matter the embedding delay $L$ the global parabolic structure was maintained.

We developed a continuous version of Graph Permutation Entropy with a novel ordinal activation function. This strictly monotonically increasing function can emphasize certain features of a given signal (spikes) or suppress undesired ones (noise) while maintaining the ordinal structure of the embedding vector and keeping the overall number of each pattern permutation the same. We showed there are many ways we can optimize the activation permutation function and the  statistic we use to extract window information. This creates a link between $CPE_G$ and deep learning by being able to optimize parameters to get a desired result. Note that this method, similar to Zanin~\cite {Zanin23} creates another graph signal  with $N$ nodes (or an $N$ length signal in the uni-variate case). This can be considered a convolution filtering operation. Thus, there is now a way to learn the optimal `filter' by exploring an infinite search space of monotonically increasing functions. In terms of further applications to machine learning or graph neural networks~\cite{gnn}, this is promising.

We showed that $CPE_G$ improves upon $PE_G$ by being able to detect sharper changes in activity and more islands of stability and considers the amplitude and overall structure in the window. We showed in real-life and synthetic data examples how $CPE_G$ can detect spikes in activity where Discrete $PE_G$ and its contrasts fail. While we only demonstrated this in the case of $m=3$ for simplicity, this is easily extended to larger embedding vectors by considering the other possible permutation patterns exactly like the discrete case. 

 Given the node-level granularity of $CPE_G$ we can now optimize at a node level opening up possibilities for more precise machine learning algorithms. Using spectral surrogate data where we used the GFT to destroy non-linearity in the signals of heartbeat data of patients, it was shown that $CPE_G$ can distinguish well between the presence and non-presence of non-linearity at a participant level better than discrete $PE_G$. We noticed that the computational time in both discrete and continuous $PE_G$ were around the same with no significant differences.

 We must address a limitation of this study. $CPE_G$, as of now and similar to deep learning, remains a relatively black-box method. While we can manipulate parameters and get an intuition of what features we are trying emphasize, we cannot tell exactly what change in input is causing what change in output. While this reduces interpretibility, it does not negate its benefits as a useful tool. This is one of the more un-desirable links to deep learning we have to integrate to currently get improved performance \cite {Goodfellow16}. However, this limitation does not preclude the interpretatibility of the contrasts, which we extended here to the graph signal case, or the fact that our formulation for graph signals enable the application of permutation entropy concepts to new kinds of data.

\section{Conclusions and Future Work}

This work sets the groundwork for the efficient application of Graph Permutation Entropy by providing a thorough extension of advances in Permutation Entropy over the past few years and translating it effectively to the graph domain. While some of these methods have been explored before, they have been constrained to the uni-variate case (or image case). The extensions to the graph domain are a first in the field.

$CPE_G$, the continuous counterpart of $PE_G$, marks a significant advancement in ordinal deep learning, extending the groundwork laid by Zanin. We introduce the ordinal activation function, enabling exploration within the infinite search space of strictly monotonically increasing functions. While strictly monotonically decreasing functions could maintain global consistency with $PE_G$ as per the in-variance theorem, their use may compromise interpretability in ordinal contrasts and granular node-level analyses due to the reversal in order. This undermines $CPE_G$'s hallmark feature of being able to maintain the distinct ordinal patterns of each node from classical discrete $PE_G$ while adding extra continuous information.

While this extension is demonstrably useful, our interpretability has dropped compared to traditional $PE$ and $PE_G$. We highly encourage further work to develop methods to better interpret $CPE_G$ results or even to integrate new methods to improve interpretability. For instance, in our study of nodal $CPE_G$ in MCI, we found that interpretability could be improved by considering $CPE_G$ alongside $PE_G$. By analyzing changes in dominant pattern in $PE_G$ (across disease), and comparing them to changes in $CPE_G$, we could detangle those changes driven by pattern, or by the amplitude aware values, leading to improved interpretability. 

This improvement in interpretability and sensitivity could be of great benefit in the study of neurological diseases such as AD. Specifically, the literature presents scattered evidence suggesting that in the pre-clinical stages of the Alzheimer's Disease continuum, there are increases in recruitment and connectivity, denoting compensatory changes, followed by massive reductions as the disease progresses. CPEG could be a tool sensitive enough to detect such transitions from hyperactivity to hypoactivity.

It is obvious that the extension to the continuous case opens up applications in machine learning due its inherently optimizable parameters. We encourage future research in this area as we believe this can be fruitful. Of particular importance is the granularity of these graph methods. We have essentially assigned a continuous value to each node in the graph in $CPE_G$, leading to potential for optimization at a node level. This can be particularly useful when attempting to emphasize certain features of graph signals. While we have explored some surrogate data, we encourage the application of more surrogate data to these graph methods to test their ability to detect further dynamics of signals and systems.

We found the results in the application of graph image contrasts fascinating, particularly in the fractal surface analysis. Further theory to explain the existence of the parabolic behaviour specific to the graph case would be beneficial. Mathematical proofs and extensions of the individual methods in this paper, such as an equivalent of the universal approximation theorem will allow a comprehensive understanding of the methods and lead to improvements. We have not explored this in this paper but we strongly believe it can be highly beneficial.

It is vital to note that this work is not merely an adaptation of previous amplitude-aware extensions of Permutation Entropy, but a novel general framework that uses the projection of embedding vectors onto a continuous ordinal activation function to get more fine-grained detail on the signal's non-linear dynamics while maintaining the simplicity of the ordinal patterns in traditional PE.

 Note that $CPE_G$ is easily transferable to uni-variate time series with just the use of traditional $PE$ instead of $PE_G$.  Overall, these novel methods should accelerate the transition of ordinal analysis to the graph domain allowing it to quickly catch up to its uni-variate counterpart. This can hopefully allow for future research to simultaneously account for the graph domain when developing ordinal analysis methods.  

\section{Acknowledgements}
JE, OR, and JSFC acknowledge a Research Project Grant (RPG-2020-158) from the Leverhulme Trust. OR is supported by the Engineering and Physical Sciences Research Council (EPSRC) Student Excellence Award (SEA) Studentship provided by the United Kingdom Research and Innovation (UKRI) council. ACC acknowledges Edinburgh University's Principle's Career Development PhD Scholarship along with Federica Guazzo for her help in pre-processing the fMRI data.


\begin{thebibliography}{10}


\bibitem{Carrasco22}
J. S. Fabila-Carrasco, C. Tan, and J. Escudero, "Permutation Entropy for Graph Signals," \textit{IEEE Trans. Signal Inf. Process. Netw.}, vol. 8, pp. 288-300, 2022, doi: 10.1109/TSIPN.2022.3167333.

\bibitem{Bandt08}
C. Bandt and B. Pompe, "Permutation Entropy: A Natural Complexity Measure for Time Series," \textit{Phys. Rev. Lett.}, vol. 88, no. 17, 174102, 2002.

\bibitem{mvperm}
J. S. Fabila Carrasco, C. Tan, and J. Escudero, "Multivariate permutation entropy, a Cartesian graph product approach," in \textit{Proc. 30th Eur. Signal Process. Conf. (EUSIPCO)}, Belgrade, Serbia, 2022, pp. 1-5.

\bibitem{Roy24}
O. Roy, Y. Moshfeghi, A. Ibanez, F. Lopera, M. A. Parra, and K. M. Smith, "FAST functional connectivity implicates P300 connectivity in working memory deficits in Alzheimer's disease," \textit{arXiv preprint arXiv:2402.18489}, 2024.

\bibitem{knn}
L. G. J. M. Voltarelli, A. A. B. Pessa, L. Zunino, R. S. Zola, E. K. Lenzi, M. Perc, and H. V. Ribeiro, "Characterizing unstructured data with the nearest neighbor permutation entropy," \textit{Chaos}, vol. 34, no. 5, 053130, 2024.

\bibitem{john24}
J. S. Fabila-Carrasco, A. Campbell-Cousins, M. A. Parra-Rodriguez, and J. Escudero, "Graph-based permutation patterns for the analysis of task-related fMRI signals on DTI networks in mild cognitive impairment," in \textit{Proc. IEEE Int. Conf. Acoust., Speech Signal Process. (ICASSP)}, 2024, pp. 2076-2080.

\bibitem{Wang23}
X. Wang, "An investigation on fractal characteristics of the superposition of fractal surfaces," Fractal Fract., vol. 7, no. 11, pp. 802, Nov. 2023, doi: 10.3390/fractalfract7110802.

\bibitem{paeng22}
S.-H. Paeng, “Curvature and entropy of a graph,” Physica A: Stat. Mech. Appl., vol. 603, pp. 127783, 2022, doi: 10.1016/j.physa.2022.127783.


\bibitem{norm}
S. Gopal, K. Patro, and K. Kumar Sahu, “Normalization: A
Preprocessing Stage.” arXiv preprint arXiv:1503.06462 (2015).


\bibitem{Zanin23}
M. Zanin, “Continuous ordinal patterns: Creating a bridge between ordinal analysis and deep learning,” Chaos, vol. 33, no. 3, p. 033114, Mar. 2023. doi: 10.1063/5.0136492

\bibitem{Dong24}
K. Dong and D. Li, “An Information Theoretic Approach to Analyze Irregularity of Graph Signals and Network Topological Changes Based on Bubble Entropy,” Fluctuation  and  Noise  Letters, 2024. doi: 10.1142/S0219477524500524

\bibitem{Bandt23}
C. Bandt, "Statistics and contrasts of order patterns in univariate time series," Chaos, vol. 33, no. 3, p. 033124, Mar. 2023. doi: 10.1063/5.0132602

\bibitem{Bandt22}
C. Bandt and K. Wittfeld, "Two new parameters for the ordinal analysis of images," Chaos, vol. 33, no. 4, p. 043124, Apr. 2023. doi: 10.1063/5.0136912

\bibitem{Bona23}
M. Bóna and J. Pantone, "Permutations avoiding sets of patterns with long monotone subsequences," Journal of Symbolic Computation, vol. 116, pp. 130-138, 2023. doi: 10.1016/j.jsc.2022.09.001

\bibitem{Rib12}
H. V. Ribeiro, L. Zunino, E. K. Lenzi, P. A. Santoro and R. S. Mendes, “Complexity-entropy causality plane as a complexity measure for two-dimensional patterns,” in PLoS One, vol. 7, no. 8, pp. e40689, Aug. 2012, doi: 10.1371/journal.pone.0040689.

\bibitem{gelu}
D. Hendrycks and K. Gimpel, “Gaussian Error Linear Units (GELUs),” arXiv preprint arXiv:1606.08415 (2016)..

\bibitem{Pirondini16}
E. Pirondini, A. Vybornova, M. Coscia and D. Van De Ville, "A Spectral Method for Generating Surrogate Graph Signals," in IEEE Signal Processing Letters, vol. 23, no. 9, pp. 1275-1278, Sept. 2016, doi: 10.1109/LSP.2016.2594072.


\bibitem{AR1}
M. A. Al-Osh and E-E. A. A. Aly, "First order autoregressive time series with negative binomial and geometric marginals," Commun. Statist.-Theory Meth., vol. 21, no. (9), pp. 2483–2492, 1992.




\bibitem{shuman}
D. I. Shuman, B. Ricaud, and P. Vandergheynst, "Vertexfrequency analysis on graphs", Applied and Computational Harmonic Analysis, vol. 40, no. 2, pp. 260-291, 2014.
 
\bibitem{temp}
B. Girault, "Stationary graph signals using an isometric graph translation". In 23rd European Signal Processing Conference (EUSIPCO), 2015, pp. 1516-1520.

\bibitem{Carbone07}
A. Carbone, "Algorithm to estimate the Hurst exponent of high-dimensional fractals," Phys Rev E Stat Nonlin Soft Matter Phys., vol. 76, no. 5 Pt 2, pp. 056703, Nov. 2007, doi: 10.1103/PhysRevE.76.056703.


\bibitem{shannon48}
C. E. Shannon, "A Mathematical Theory of Communication," Bell System Technical Journal, vol. 27, pp. 379-423, 1948, doi: 10.1002/j.1538-7305.1948.tb01338.x

\bibitem {May76}
R. M. May, "Simple mathematical models with very complicated dynamics," Nature, vol. 261, no. 5560, pp. 459, 1976.


\bibitem {kl}

Joyce, J.M. (2011). Kullback-Leibler Divergence. In: Lovric, M. (eds) International Encyclopedia of Statistical Science. Springer, Berlin, Heidelberg. 


\bibitem {heartbeat}
Z. Chen, Y. Li, H. Liang, and J. Yu, "Improved permutation entropy for measuring complexity of time series under noisy condition", Complexity, pp. 1-12, 2019.
\bibitem {heartbeat1}
C. Bian, C. Qin, Q. D. Ma, and Q. Shen, "Modified permutation-entropy analysis of heartbeat dynamics", Physical Review E, vol. 85, no. 2, pp. 021906, 2012.

\bibitem {adam}
D. P. Kingma and J. Ba, "Adam: A method for stochastic optimization," arXiv preprint arXiv:1412.6980, 2014.



\bibitem {Gold76}
A. L. Goldberger et al., "PhysioBank, PhysioToolkit, and PhysioNet: Components of a New Research Resource for Complex Physiologic Signals," Circulation, vol. 101, no. 23, pp. e215-e220, June 2000. 

\bibitem {alz}
P. Scheltens et al., "Alzheimer's disease," Lancet, vol. 397, no. 10284, pp. 1577-1590, Apr. 2021, doi: 10.1016/S0140-6736(20)32205-4.


\bibitem {mci}
S. Gauthier et al., "Mild cognitive impairment," Lancet, vol. 367, no. 9518, pp. 1262-1270, Apr. 2006, doi: 10.1016/S0140-6736(06)68542-5.


\bibitem {mario}
M. A. Parra et al., "Memory markers in the continuum of the Alzheimer’s clinical syndrome," Alz Res Therapy, vol. 14, no. 142, 2022, doi: 10.1186/s13195-022-01082-9.

\bibitem{brain}
M. A. Parra et al., "Visual short-term memory binding deficits in familial Alzheimer’s disease," Brain, vol. 133, no. 9, pp. 2702–2713, Sept. 2010.


\bibitem{nimage}
C. R. Buchanan et al., "Test–retest reliability of structural brain networks from diffusion MRI," NeuroImage, vol. 86, pp. 231-243, 2014, doi: 10.1016/j.neuroimage.2013.09.054.

\bibitem{bviewer}

M. Xia, J. Wang, and Y. He, “BrainNet Viewer: A Network Visualization Tool for Human Brain Connectomics,” PLoS ONE, vol. 8, no. 7, pp. e68910, 2013, doi: 10.1371/journal.pone.0068910

\bibitem {Kylberg16}
G. Kylberg, and I. Sintorn, "On the influence of interpolation method on rotation invariance in texture recognition," J Image Video Proc., vol. 2016, no. 17, 2016, doi: 10.1186/s13640-016-0117-6

\bibitem {midpoint}
W. Lau, A. Erramilli, J. L. Wang and W. Willinger, "Self-similar traffic generation: the random midpoint displacement algorithm and its properties," Proceedings IEEE International Conference on Communications ICC '95, Seattle, WA, USA, 1995, pp. 466-472 vol.1, doi: 10.1109/ICC.1995.525213.

\bibitem {rosso}
 O. Rosso et al., "Distinguishing noise from chaos," Phys. Rev. Lett., vol. 99, no. 154102, 2007.

\bibitem {Ortega18}
A. Ortega et al., "Graph Signal Processing: Overview, Challenges, and Applications," Proceedings of the IEEE, vol. 106, no. 5, pp. 808-828, May 2018, doi: 10.1109/JPROC.2018.2820126.

\bibitem {Medgalia17}
J. D. Medgalia, "Graph Theoretic Analysis of Resting State Functional MR Imaging," Neuroimaging Clin N Am., vol. 27, no. 4, pp. 593-607, Nov. 2017, doi: 10.1016/j.nic.2017.06.008.

\bibitem {gnn}

J. Zhou et al., “Graph neural networks: A review of methods and applications,” AI Open, vol. 1, pp. 57-81, 2020, doi: 10.1016/j.aiopen.2021.01.001.


\bibitem {Goodfellow16}


I. Goodfellow, Y. Bengio and A. Courville, Deep Learning, Cambridge, MA: MIT Press, 2016.
\bibitem {didic11}

M. Didic et al., "Which memory system is impaired first in Alzheimer's disease?," J Alzheimers Dis, vol. 27, no. 1, pp. 11-22, 2011.

\bibitem {insel20}

P. S. Insel et al., "Neuroanatomical spread of amyloid beta and tau in Alzheimer’s disease: implications for primary prevention," Trends in Neurosciences, vol. 2, no. 1, fcaa007, 2020.


\bibitem {klos23}
F. Klostermann and H. O. Tiedt, "Thalamic and basal ganglia involvement in language-related functions," Brain Structure and Function, vol. 54, pp. 101323, 2023.


\bibitem {mayes07}
A. Mayes, D. Montaldi, and E. Migo, "Associative memory and the medial temporal lobes," Trends Cogn Sci, vol. 11, no. 3, pp. 126-135, 2007.

\bibitem {parra22}
M. A. Parra, "Barriers to Effective Memory Assessments for Alzheimer’s Disease," J Alzheimers Dis, 2022. doi: 10.3233/jad-215445

\bibitem{parra14}
M. A. Parra, S. Della Sala, R. H. Logie, and A. M. Morcom, "Neural correlates of shape-color binding in visual working memory," Neuropsychologia, vol. 52, no. 0, pp. 27-36, Jan. 2014. doi: 10.1016/j.neuropsychologia.2013.11.008

\bibitem{parra13}
M. A. Parra, V. Pattan, D. Wong, A. Beaglehole, J. Lonie, H. I. Wan, et al., "Medial temporal lobe function during emotional memory in early Alzheimer's disease, mild cognitive impairment and healthy ageing: an fMRI study," BMC Psychiatry, vol. 13, p. 76, 2013. doi: 10.1186/1471-244X-13-76

\bibitem{star10}
B. P. Staresina and L. Davachi, "Object unitization and associative memory formation are supported by distinct brain regions," J. Neurosci., vol. 30, no. 29, pp. 9890-9897, Jul. 2010.

\bibitem{valdes20}
M. C. Valdés Hernández, R. Clark, S.-H. Wang, F. Guazzo, C. Calia, V. Pattan, et al., "The striatum, the hippocampus, and short-term memory binding: Volumetric analysis of the subcortical grey matter's role in mild cognitive impairment," NeuroImage: Clin., vol. 25, p. 102158, 2020. doi: 10.1016/j.nicl.2019.102158

\bibitem{zach14}
V. Zachariou, R. Klatzky, and M. Behrmann, "Ventral and dorsal visual stream contributions to the perception of object shape and object location," J. Cogn. Neurosci., vol. 26, no. 1, pp. 189-209, Jan. 2014.


\bibitem {minwoo}
M. Lee, "Early warning detection of thermoacoustic instability using three-dimensional complexity-entropy causality space," Exp. Therm. Fluid Sci., vol. 130, p. 110517, 2022. doi: 10.1016/j.expthermflusci.2021.110517




\end{thebibliography}
\end{document}